\documentclass[%
 reprint,
 amsmath,amssymb,
 aps,
]{revtex4-2}
\usepackage{comment}
\usepackage{graphicx}% Include figure files
\usepackage{dcolumn}% Align table columns on decimal point
\usepackage{bm}% bold math
\usepackage{qcircuit} % quantum circuits 
\usepackage{caption}
\usepackage{subcaption}
\usepackage{braket}   % bra-ket 
\usepackage{lmodern}
\usepackage{relsize} % resizing the sum

\usepackage{verbatim,multirow}
\usepackage{multirow}
\usepackage[utf8]{inputenc}
\usepackage[T1]{fontenc}

\usepackage{soul}

\newcommand{\cA}{\mathcal{A}}

\newcommand{\cC}{\mathcal{C}}

\newcommand{\cE}{\mathcal{E}}
\newcommand{\cF}{\mathcal{F}}

\newcommand{\cO}{\mathcal{O}}

\newcommand{\cR}{\mathcal{R}}
\newcommand{\cS}{\mathcal{S}}

\usepackage{graphicx,color,appendix,physics,wasysym,xcolor,xparse}
\usepackage{setspace,appendix} 
\usepackage{amsmath,amssymb,amsfonts,mathtools}
\usepackage{bm,bbm,euscript,braket}
\usepackage{hyperref}
\hypersetup{colorlinks,linkcolor={red},citecolor={blue},urlcolor={blue}}

\newtheorem{Theorem}{Theorem}
\newtheorem{lemma}{Lemma}

\newenvironment{proof}{{\bf Proof:}}{\hfill$\square$}

\usepackage{tikz}

\begin{document}

\title{Noise-adapted qudit codes for amplitude-damping noise}% Force line breaks with \\
%\thanks{A footnote to the article title}%

 %\altaffiliation[Also at ]{}%Lines break automatically or can be forced with \\
\author{Sourav Dutta}
\thanks{Both the authors contributed equally}
 %\email{sourav@physics.iitm.ac.in}
 \author{Debjyoti Biswas}
\thanks{Both the authors contributed equally}
%\email{d4bj@physics.iitm.ac.in}
 \author{Prabha Mandayam}
 %\homepage{prabhamd@physics.iitm.ac.in}
\affiliation{%
 Department of Physics, Indian Institute of Technology Madras\\
 Chennai, India - 600036
}%
\affiliation{
 Center for Quantum Information, Communication, and Computing, IIT Madras
}%
%\affil[*]{Authors contributed equally}
%\collaboration{MUSO Collaboration}%\noaffiliation

%\date{\today}% It is always \today, today,
             %  but any date may be explicitly specified
\begin{abstract}
 %The biggest challenge in realizing a robust and scalable quantum computing device is the decoherence of quantum states. 
Quantum error correction (QEC) plays a critical role in preventing information loss in quantum systems and provides a framework for reliable quantum computation. Identifying quantum codes with nice code parameters for physically motivated noise models remains an interesting challenge. While past work has primarily focused on qubit codes, here we identify a $[4,1]$ qudit error correcting code tailored to protect against amplitude-damping noise. We show that this four-qudit code satisfies the error correction conditions for all single-qudit and a few two-qudit damping errors up to the leading order in the damping parameter $\gamma$. We devise a protocol to extract syndromes that unambiguously identify this set of errors, leading to a noise-adapted recovery scheme that achieves a fidelity loss of $\cO(\gamma^{2})$. For the $d=2$ case, our QEC scheme is identical to the known example of the $4$-qubit code and the associated syndrome-based recovery.
We also assess the performance of this code using the Petz recovery map and note some interesting deviations from the qubit case. Finally, we generalize this construction to a family of $[2M+2, M]$ qudit codes that can approximately correct all the single-qudit and a few two-qudit amplitude-damping errors.
\end{abstract}

\maketitle

\section{Introduction}
%Why quantum computing

Quantum computers have the remarkable potential to speed
up complex computational tasks by harnessing the princi-
ples of quantum mechanics. While most quantum devices
in operation have thus far used qubits—two-dimensional
quantum systems—as their fundamental building blocks of
information, there is a growing number of quantum devices
built upon higher-dimensional quantum systems \cite{qudits_comp1, qudit_comp2, qudit_comp3, qudit_jaynes_cummings, qudit1, qudit2, quditvqe}. These
d-level quantum systems—often called qudits—typically
have a larger information capacity and admit more complex
computations, making them a promising avenue for scaling up
quantum algorithms and achieving enhanced computational
power  \cite{qudit1, qudit2, quditvqe, qutrit}.

However, while moving to higher-dimensional systems
helps increase the efficiency of quantum algorithms, these
systems are also more susceptible to errors than qubits. Since
qudit systems also have more accessible levels, they are
more difficult to control compared to qubits. For example,
transmon-based qudit systems predominantly undergo multi-
level decay errors, whereas transmon qubits are mostly fragile
to single-level-decay error.

Quantum error correction \cite{terhal_qec} provides a framework by which one can systematically deal with noise in quantum systems and improve the reliability and accuracy of quantum computing devices. Conventional quantum error correction (QEC) codes, such as the surface codes~\cite{fowler2012} or stabilizer codes~\cite{gottesman} in general, were primarily designed to correct for arbitrary, independently occurring Pauli errors. These general-purpose QEC codes are governed by rigid constraints on the minimal resources required for achieving error correction. For example, the smallest QEC codes have a block length
of five, requiring at least five qudits to encode a single qudit
to correct for any arbitrary error on a single system~\cite{qudit_hamming, singleton}.

Physical realizations of qudits are often dominated by certain specific noise processes with their own error characteristics. For instance, the decoherence mechanism in superconducting qubits is dominated by amplitude-damping noise over other noise processes~\cite{sourenoise}. In such cases, it is possible to devise encoding schemes that aim to correct for specific errors and hence require fewer resources compared to general-purpose QEC codes~\cite{Leung}. This then leads to the idea of \emph{noise-adapted} quantum error correction schemes that are tailored to deal with specific noise models \cite{barnum2002, prabha, fletcher_channel}. Since noise-adapted codes achieve a degree of protection comparable to general-purpose QEC codes while being more resource-efficient, they are potentially good candidates for implementing on today's  NISQ (Noisy Intermediate-Scale Quantum) devices~\cite{preskill2018}. Noise-adapted QEC protocols have been developed for biased noise models, where some Pauli errors are more likely to occur than others~\cite{preskill_biased,bias_camp, liang_jiang_biased, s_puri_bias}.

%literature review

Going beyond Pauli noise channels, the amplitude-damping channel is an important example of a  noise model that has been well studied from the perspective of noise-adapted QEC. Apart from the $4$-qubit codes that correct for single-qubit damping errors~\cite{fletcher_channel, jayashankar2020finding}, other examples of amplitude-damping codes constructed in the literature include qudit codes that correct for single-qubit damping errors~\cite{Addnoise, wasilewski} and permutation-invariant codes that use $t^{2}$ qubits to approximately correct for $t$ damping errors~\cite{permutation_AD}. Extending the idea of encoding qubits into qudits, the class of bosonic codes, including binomial codes \cite{bin_code_vvalbert, bosonics_vva} and cat codes \cite{cat_1, cat_2,cat_3,s_puri, s_puri_1, bosonics_vva, review_bose_code} are primarily designed to correct for single-damping errors which arise due to single photon loss processes.
Qudit stabilizer codes that correct for amplitude-damping noise on qudits have also been developed from classical asymmetric codes~\cite{mg}.

In this article, we first identify a four-qudit code that can
protect against amplitude-damping noise. We show that this
code satisfies a set of approximate QEC conditions for all single-qudit damping errors and a few two-qudit damping errors. We observe that the well-known four-qubit Leung code \cite{Leung} is a special case of this code construction when the local
dimension is two.

Secondly, we propose a noise-adapted recovery procedure for the four-qudit code, tailored to detect and correct for amplitude-damping noise, via a syndrome-based approach. 
Similar to the four-qubit amplitude damping code, one needs extra measurements~\cite{Aj} to identify the errors uniquely when the local dimension $d$ is less than or equal to four. 
Our findings demonstrate that two stabilizer measurements are indeed adequate for correcting single-qudit damping errors up to $\cO(\gamma^2)$ when $d\geq5$. Finally, we also present a quantum circuit implementation of our syndrome-based recovery and compare its performance with the near-optimal recovery map known as the Petz recovery\cite{barnum2002,junge2018,prabha}.
It is interesting to observe that, in contrast to the qubit scenario, the syndrome-based recovery outperforms the Petz recovery in the qutrit case.

The rest of the paper is organised as follows. In Sec.~\ref{sec:prelim}, we briefly review the structure of amplitude-damping noise for qudits and discuss the approximate quantum error correction conditions which form the basis for the code construction.
In Sec.~\ref{sec:4quditcode}, we describe the four-qudit approximate QEC code for amplitude damping noise and demonstrate a syndrome-based recovery scheme for this code in Sec.~\ref{sec:stabilizer for the code}. We quantify the performance of the qutrit code with different recovery operations in Sec.~ \ref{sec:qutrit}. In Sec.~\ref{section:generalization}, we generalize the four-qudit code to a family of $[2M+2, M]$ qudit error correcting codes for correcting amplitude-damping noise up to $\cO(\gamma)$. 
Finally, we conclude in Sec.~\ref{sec:conclusion} with a summary and future directions.

\section{Preliminaries}\label{sec:prelim}
\subsection{Approximate quantum error correction}
Quantum states are fragile and susceptible to noise arising due to unwanted interactions with a bath or environment. The effect of noise can be modelled as a completely positive trace preserving (CPTP) map $\cE$ -- commonly referred to as a \emph{quantum channel} -- whose action on any state $\rho$ is described by a set of \emph{Kraus operators} $\{E_{k}\}$, such that $\cE(\rho) = \sum_{k}E_{k}\rho E_{k}^{\dagger}$. The map $\cE$ is guaranteed to be trace-preserving so long as the operators $E_{k}$ satisfy $\sum_{k}E_{k}^{\dagger} E_{k} = I$.

The central idea in quantum error correction (QEC) is to encode the information contained in the physical  Hilbert space into a subspace of a larger Hilbert space, known as the codespace. In order to be able to correct for a noise channel $\mathcal{E}$ represented by a set of $N$ Kraus operators $\{E_k\}_{a=1}^N$, the codespace must satisfy specific algebraic conditions known as the Knill-Laflamme conditions~\cite{KLCondition}, stated as follows. A codespace $\mathcal{C}$ defined as the span of the quantum states $\{|i_L\rangle\}_{i=1}^{dim(\mathcal{C})}$ corrects the errors associated with operators ${E_k}$ if and only if,
\begin{equation}
    \langle i_L|  E_k^{\dagger} E_l |j_L \rangle = c_{kl} \delta_{ij} \qquad \forall~~ k, l, \label{a1}
\end{equation}
for every pair of codewords $|i_{L}\rangle, |j_{L}\rangle$. Here, $c_{kl}$ are elements of a Hermitian matrix, and are independent of $i, j$. The independence of $c_{kl}$ from $i$ and $j$ ensures that the errors are simply unitary deformations of the codespace, whereas the Kronecker delta in Eq.~\eqref{a1} implies that orthogonality of the error subspaces associated with distinct errors~\cite{nielsen}.

%Approximate QEC
More generally, one may consider relaxing the perfect QEC conditions, leading to approximate quantum error correction (AQEC) codes~\cite{Leung, beny, prabha}. There are a few variants of AQEC conditions proposed in the literature, depending on how the exact QEC conditions are relaxed and the figure of merit used to benchmark the performace of the codes. The specific form of the AQEC conditions that we use here is due to~\cite{beny}, where the Knill-Laflamme conditions in Eq.~\eqref{a1} are modified as follows. A quantum code $\cC$ satisfying 
\begin{equation}
    \langle i_L|  E_k^{\dagger} E_l |j_L \rangle = c_{kl} \delta_{ij} + \langle i_L|  B_{kl} |j_L \rangle\qquad \forall~~ k, l, \label{eq:appx_cond}
\end{equation}
for every pair of codewords $|i_{L}\rangle, |j_{L}\rangle$, can correct the errors $\{E_{k}\}$ up to $\cO(\epsilon^t)$ if and only if the perturbative term in Eq.~\eqref{eq:appx_cond} is bounded as
\begin{equation}\label{eq:beta}
    \langle i_L|  B_{kl} |j_L \rangle \leq \cO(\epsilon^{t+1}).
\end{equation} 
Clearly, the code $\cC$ becomes a \emph{perfect} QEC code for the noise channel $\cE$ when $\langle i_L|  B_{kl} |j_L \rangle$ in \eqref{eq:appx_cond} goes to zero. Indeed, if $\bra{i_L} B_{kl}\ket{j_L} $ vanishes for all $i\neq j$, then the orthogonality condition in Eq.~\eqref{a1} is satisfied. In that case, there exists a trace non-increasing completely positive map, that can approximately correct the effect of the noise channel $\cE$ on the code space $\cC$~\cite{Leung,cafaro}.

\subsection{\label{sec: AmpDamp}Amplitude-damping noise}
The amplitude-damping (AD) noise channel models the phenomenon of energy dissipation in open quantum systems~\cite{nielsen}. It is known to be one of the most commonly occurring and oftentimes dominant noise process for several physical realizations of qubits \cite{dom_ampd}. 
AD errors arise due to the loss of energy from the system to its surrounding environment, leading to loss of coherence as the probability of finding the qubit in the excited state diminishes. 
The strength of AD noise is characterized by $\gamma$, the probability of the qubit getting \emph{damped}, that is, decaying from the ground state from the excited state. Correspondingly, $1-\gamma$ is the probability that the qubit remains in the excited state.

At zero temperature, this damping process for qubits is characterized by the following pair of Kraus operators, which describe a single-qubit amplitude-damping channel. 
\begin{align}\label{eq:Ad_qubit}
A_0&= \begin{pmatrix} 1 & 0 \\ 0 & \sqrt{1-\gamma}\end{pmatrix} \qquad A_1 =  \begin{pmatrix} 0 & \sqrt{\gamma} \\ 0 & 0\end{pmatrix} .
\end{align}
Going beyond qubits, energy dissipation in qudit systems ($d$-level quantum systems), can be described via a qudit amplitude-damping channel $\cA$, with $d$ Kraus operators $A_k$  {($0\leq k \leq d-1$)} of the form \cite{mg},
\begin{align}\label{eq:ad_multi}
A_k &= \sum\limits_{r=k}^{d-1}\, \sqrt{r \choose k }\sqrt{ (1-\gamma)^{r-k} \gamma^k}\ket{r-k}\bra{r},
\end{align}
where $\gamma$ is now the probability for a single damping event to occur. The quantum state $\ket{r}$ represents a system in the $r^{\rm th}$ excited state. For a $d$-level quantum system, the index $k$ runs from $0$ to $d-1$. 
The operator $A_{k}$ thus describes a $k$-level damping event that occurs with probability $\cO(\gamma^{k})$. In the rest of the article, we will refer to $A_{k}$ as a $k$-damping error with a corresponding damping strength of $\cO(\gamma^{k})$.

A schematic representation of the qubit and qutrit amplitude-damping channels is shown in Fig. \ref{fig:qudit_damping}. In the qutrit case, a single damping event causes the system to decay from the state $\ket{2}$ to the state $\ket{1}$,  without decaying any further. Thus the probability for this process is $\gamma(1-\gamma)$ up to a constant factor. We can thus obtain the Kraus operators for a qutrit amplitude-damping channel by substituting $d=3$ in Eq.~\eqref{eq:ad_multi}, as, 
\begin{align}\label{eq:qtr_ad}
    A_0 &= \ket{0}\bra{0}+\sqrt{1-\gamma} \ket{1}\bra{1} +(1-\gamma)\ket{2}\bra{2}\\
      A_1 &= \sqrt{\gamma}\ket{0}\bra{1}+\sqrt{2\gamma(1-\gamma)} \ket{1}\bra{2}\\
      A_2 &= \gamma \ket{0}\bra{2},
\end{align}
where $A_{0}$ is the no-damping error, $A_{1}$ corresponds to the single-damping error and $A_{2}$ is the two-damping error operator.

\begin{figure}[t]
    \centering
    \includegraphics[width=1\columnwidth]{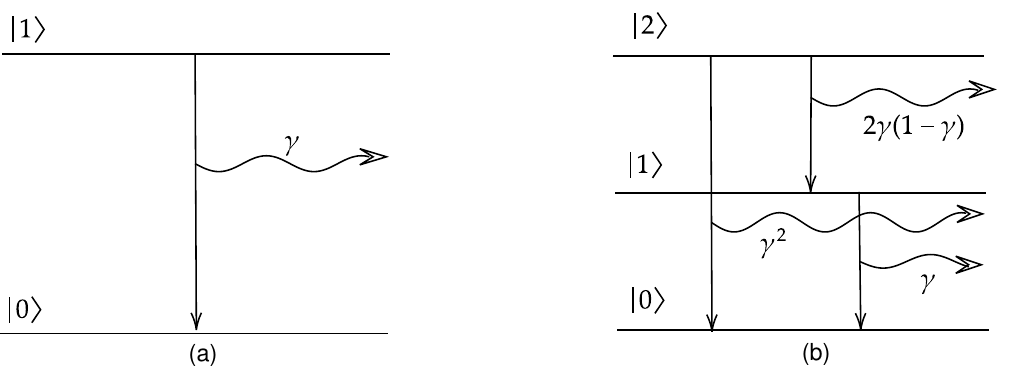}
    \caption{(a) Amplitude-damping process for a qubit system, where only one possible damping from $\ket{1}$ to $\ket{0}$ takes place with probability $\gamma$. (b) Amplitude-damping process for a qutrit system, where damping from $\ket{2}$ to $\ket{0}$ takes place with probability $\gamma^2$. The probability of damping from $\ket{2}$ to $\ket{1}$ and $\ket{1}$ to $\ket{0}$ takes place with probability $2\gamma(1-\gamma)$ and $\gamma$ respectively.  }
    \label{fig:qudit_damping}
\end{figure}

 In this work, we aim to construct quantum codes that can protect qudits from single and multi-level damping errors. It is evident that whenever $\gamma < 0.5$, which holds true for most of the realistic cases, the probability of no-damping and single-level damping is higher than multi-level damping. This implies that the single-level damping errors affect the fidelity of a quantum code the most.

\subsection{Noise-adapted recovery}\label{sec:NA_rec}

Before proceeding to  {the} code construction, we briefly review the known noise-adapted recovery schemes tailored to specific codes and noise processes. In standard QEC, the recovery operation is simply a unitary gate -- typically, a Pauli operator -- depending on the syndrome bits extracted. On the other hand, noise-adapted QEC requires a non-trivial recovery operation, which could, in general, be a quantum channel. Given a code $\cC$ and a noise channel $\cE$, there are three known approaches to constructing a noise-adapted recovery scheme.  

\begin{enumerate}
    \item \textit{Petz recovery:} This approach to noise-adapted recovery is based on a universal and near-optimal recovery channel, often known as the Petz map~\cite{barnum2002, prabha}. The Petz recovery specific to a code $\cC$ and adapted to a noise channel $\cE$ with Kraus operators $\{E_i\}^N_{k=1}$, is the CPTP map $\cR_{P}$ whose Kraus operators are given by,
    \begin{align}
        R^{(P)}_{k} = P E^{\dag}_k \cE(P)^{-1/2}, \label{eq:petz}
    \end{align}
    where $P$ is the projector onto the code space $\cC$. Recent works~\cite{gilyen2022, Biswas} have shown that it is possible to construct circuit implementations of the Petz recovery channel for arbitrary codes and noise processes.
    \item \textit{Leung recovery:} For a quantum code $\cC$ and a noise channel $\cE$ with Kraus operators $\{E_i\}^N_{k=1}$, the recovery map due to Leung \emph{et al..}~\cite{Leung} is constructed using the following set of Kraus operators. 
    \begin{equation}\label{eq:leung_kraus}
        R^{(L)}_k = P U_k^{\dag},
    \end{equation}
    where $U_{k}$ is the unitary operator obtained via polar decomposition of the operator $E_kP$ and $P$ is the projector onto the codespace. The map $\cR_{L}$ is trace non-increasing provided the error operators map the codespace to orthogonal subspaces. This means that an additional constraint $\bra{i_L}B_{kl}\ket{j_L}\propto \delta_{kl}$ should be imposed on the approximate QEC conditions in Eq.~\eqref{eq:appx_cond} for  $\cR_{L}$ to correspond to a physically realisable recovery map. Putting in this additional orthogonality constraint, it is easy to check that the set of operators $\{R^{(L)}_k\}$ along with the additional operator $P_E = I -\sum\limits_{k} R^{(L)\dag}_k R^{(L)}_k$ define a CPTP map corresponding to a valid recovery channel. 
    
    \item \textit{Cafaro recovery:} Given a code $\cC$ and a noise channel $\cE$ with Kraus operators $\{E_{k}\}$, the recovery map due to Cafaro \emph{et al..}~\cite{cafaro} is defined by the Kraus operators,
    \begin{align}\label{eq:cafaro_kraus}
        R^{(C)}_k &= \sum\limits_{i}\frac{\ket{i_L}\bra{i_L}E_k^{\dag}}{\sqrt{\bra{i_l}E_k^{\dag}E_k\ket{i_L}}}.
    \end{align} 
    This map is trace non-increasing if and only if $\bra{i_L} B_{kl} \ket{j_L} \propto \delta_{kl} \delta_{ij}$ in Eq.~\eqref{eq:appx_cond}. Note that this is a slightly stronger condition than the one required for the Leung recovery, since it has the additional $\delta_{ij}$ term. Similar to the Leung case, once the orthogonality condition is satisfied, the full CPTP recovery map is defined as $\{ R^{(C)}_k, I -\sum\limits_{k} R^{(C)\dag}_k R^{(C)}_k\}$.
\end{enumerate}

We can thus establish a hierarchy of noise-adapted recovery maps based on the constraints necessary to ensure that the maps are physical (completely positive and trace-preserving). The Petz recovery map in Eq.~\eqref{eq:petz} is a physical map for any noise channel $\cE$ and code $\cC$. On the other hand, the recovery maps defined using the operators in Eqs.~\eqref{eq:leung_kraus} and~\eqref{eq:cafaro_kraus} are physical if and only if the corresponding code and noise channel satisfy certain algebraic conditions. Furthermore, these two recovery maps turn out to be identical for codes and noise channels satisfying a specific form of the approximate QEC condition. This is formally stated in the following Lemma and proved in Appendix \ref{sec:app_leung}.

\begin{lemma}\label{lem: equiv}
Consider a noise channel $\cE$ and codespace $\cC$ with codewords $\{|i_{L}\rangle, i=1,2,\ldots, d\}$ that satisfy the following form of the AQEC conditions,  
\begin{equation}
    \langle i_L|  E_k^{\dagger} E_l |j_L \rangle = \left( c_{kk} \delta_{ij} + \langle i_L|  B_{kk} |j_L \rangle \right) \delta_{kl} ,\label{eq:aqec_leung)}
\end{equation}
for every pair of Kraus operators $E_{k}, E_{l}$. Then, the Leung recovery in Eq.~\eqref{eq:leung_kraus} and the Cafaro recovery in Eq.~\eqref{eq:cafaro_kraus} corresponding to this code and noise channel are identical if and only if, 
\begin{equation}
\langle i_L| E_{k}^{\dag} E_k |j_L \rangle \propto \delta_{ij}, \; \forall, i,j \in [1,d].\label{eq:aqec_orth}
\end{equation}
\end{lemma}

\subsection{Fidelity measures}

The performance of a QEC protocol is typically quantified by one of the fidelity measures that quantify how close the initial encoded state is to the recovered state. Here, we work with the fidelity function $F$ based on the Bures metric, which is defined between a pure state $|\psi\rangle$ and a mixed state $\rho$ as~\cite{nielsen},
\begin{equation}
    F^{}(|\psi\rangle, \rho) = \langle \psi | \rho |\psi\rangle . \label{eq:fidelity}
\end{equation}
Since the fidelity is a function of the initial state, one may quantify the performance of a QEC scheme over the entire codespace is using the \emph{worst-case fidelity}, which is obtained by minimising the fidelity over all possible logical states. The worst-case fidelity for a noise channel $\cE$ and recovery channel $\cR$ acting on a codespace $\cC$ is defined as\cite{nielsen}
\begin{equation}
    \cF_{wc} = \min_{\ket{\psi} \in \cC} \bra{\psi} \cR \circ \cE (\ket{\psi} \bra{\psi}) \ket{\psi} \label{eq: wcf}
\end{equation}

Alternately, one could quantify the performance of a QEC code is by using the entanglement fidelity, defined as~\cite{nielsen},
\begin{equation}
    \cF_{ent} = \bra{\Psi} (\cR \circ \cE)\otimes I (\ket{\Psi}\bra{\Psi})\ket{\Psi}, \label{eq:ent_fid}
\end{equation}
where $\ket{\Psi}$ is a purification  of the maximally mixed state on the codespace $\rho_{L} = \frac{1}{d} \sum_{i=0}^{d-1}\ket{i_L}\bra{i_L}$. 
By evaluating the right hand side of Eq.~\eqref{eq:ent_fid}, we get the following useful expression for the entanglement fidelity.
\begin{equation}
    \cF_{ent} = \frac{1}{(\textrm{dim}(\cC))^2} \sum_{jk}\left(|\Tr_{\cC}R_j E_k|^2\right).
\end{equation}
Since the state $\ket{\Psi}$ is a purification of the maximally mixed state $\rho_{L}$ on the codespace, one may view the entanglement fidelity as a measure of the entanglement-preserving capability of a given QEC scheme.

\section{Approximate four-qudit code} \label{sec:4quditcode}
We first identify an approximate four-qu\emph{d}it quantum error correcting code for arbitrary values of $d$, which can correct all single-qudit and a few multi-qudit errors. To correct the errors due to the noise models discussed in Sec.~\ref{sec: AmpDamp},  the $4$-qudit codespace is constructed as the span of a set of $d$ codewords, defined as follows.
\begin{align}\label{eq:4qdt_code}
    \ket{m_L}&=\frac{1}{\sqrt{d}} \sum\limits_{i=0}^{d-1}\, \ket{i}_1 \ket{i}_2 \ket{(i+m)_d}_3 \ket{(i+m)_d}_4 . \
\end{align}
Here, $ m \in \lbrace 0, d-1\rbrace$ and the subscript $d$ indicates addition modulo $d$. 
\begin{figure}[t] 
\centering
\includegraphics[width=0.8\columnwidth]{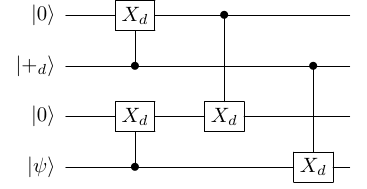}
      \caption{Encoding circuit for the approximate four-qudit code. The controlled-$X_d$ gate is defined in Eq.~\eqref{eq:cnotd} and  $\ket{+_d}= \frac{1}{\sqrt{d}} [\ket{0} +\ket{1} + \cdots +\ket{d-1}]$}.  
     \label{fig:enc_qtrt}
 \end{figure}
The encoding circuit for this qudit code is shown in Fig. \ref{fig:enc_qtrt}. The gate $X_{d}$ in the encoding circuit is simply the Pauli $X$ operator on a qudit system, given by,
\begin{equation}
    X_d = \sum_{k=0}^{d-1} \ket{(k+1)_d} \bra{k} . \label{eq:xd}
\end{equation}
This is simply a shift operator on the discrete set of states $\{|0\rangle, |1\rangle, \ldots, |d-1\rangle\}$. Correspondingly, the phase operator on a qudit system is defined as,
\begin{equation}
    Z_d = \sum_{k=0}^{d-1} e^{\frac{2\pi i k}{d}}\ket{k} \bra{k}. \label{eq:Zd}
\end{equation}
The controlled-$X_{d}$ gate, which is the $d$-dimensional analogue of the \textsc{cnot} gate is then defined as,
\begin{align}
    CX_d=\begin{bmatrix}
    \mathbf{I}_d & 0 & \cdots & 0 \\
    0 & \mathbf{X}_d & \cdots & 0 \\
    \vdots & \vdots & \ddots  & \vdots \\
    0 & 0 & \cdots & \mathbf{X}^{d-1}_d \end{bmatrix}\label{eq:cnotd}
\end{align}
Also, the state $\ket{+_d}$ is simply an equal superposition of the single-qudit basis states in the standard or computational basis. 

From  Eq.~\eqref{eq:4qdt_code}, it is evident that for $d=2$,  {this} code exactly matches the four-qubit code \cite{Leung}, tailored for qubit amplitude-damping noise. 
\begin{align}
    \ket{0_L}&= \frac{1}{\sqrt{2}}[\ket{0000}+\ket{1111}] \nonumber \\
    \ket{1_L}&= \frac{1}{\sqrt{2}}[\ket{0011}+\ket{1100}]
\end{align}
For $d\geq 3$, this code construction yields code parameters that have not been explicitly studied in the context of the amplitude-damping channel~\cite{mg} in previous constructions of stabilizer qudit codes.  We note in passing that there are
interesting connections between the $[4,1]_d$ code and the
distance-two surface codes without periodic boundary conditions. This connection is briefly explained in Appendix \ref{appendix:D}. Finally, we note that the code in Eq.~\eqref{eq:4qdt_code} can be realised as a subcode of the $[[4,2,2]]_{G}$ group-qudit code \cite{group_code}, where the group $G$ in our case is $\mathbb{Z}_d$.

We now show that similar to the $4$-qubit code, this $4$-qudit code satisfies the Knill-Laflamme condition in Eq.~\eqref{a1} for amplitude-damping noise, up to the first order in the decay strength $\gamma$. 
Specifically, we show that the codewords defined in Eq.~\eqref{eq:4qdt_code} satisfy the approximate QEC conditions in Eq.~\eqref{eq:appx_cond} for all single-qudit and some two-qudit errors, with the perturbation terms in Eq.~\eqref{eq:beta} being of $\cO(\gamma^{2})$. Note that the four-qudit amplitude-damping channel has Kraus operators of the form $A_{ijkl} = A_i \otimes A_j \otimes A_k \otimes A_l$, with the single-qudit error operators $A_{k}$ defined in Eq.~\eqref{eq:ad_multi}. The error $A_{ijkl}$ occurs with probability $\cO(\gamma^{i+j+k+l})$.

\begin{Theorem}\label{thm:qudit_qec}
The four-qudit code defined as the span of the codewords in Eq.~\eqref{eq:4qdt_code} satisfies the Knill-Laflamme conditions in Eq.~\eqref{a1} up to $\cO(\gamma)$, for the errors in the set 
   \begin{eqnarray}
       \cA_{\rm corr} &=& \{ \{A_{x000}, A_{0x00}, A_{00x0}, A_{000x} \; x\in \{1,\ldots, d-1\}\}, \nonumber \\  
       && A_{0000}, A_{1010}, A_{0101}, A_{0110}, A_{1001}\},
  \nonumber \end{eqnarray} 
comprising \emph{all} single-qudit errors and some two-qudit errors. Here, $A_{x000} = A_{x}\otimes A_{0}\otimes A_{0}\otimes A_{0}$ denotes a single-qudit $x$-damping error with a damping strength $\cO(\gamma^{x})$, whereas, $A_{1010} = A_{1}\otimes A_{0}\otimes A_{1}\otimes A_{0}$ denotes a two-qudit single-damping error of damping strength $\cO(\gamma^{2})$.   
\end{Theorem}

\begin{proof} 
We merely outline the steps of our proof here and leave the details in the Appendix~\ref{sec:proof}.

To verify that the $4$-qudit code satisfies the Knill-Laflamme conditions up to $\cO(\gamma)$, we first explicitly evaluate the action of the single-qudit damping errors on the codewords, as given in Eqs.~\eqref{eq:3aa1}-~\eqref{eq:3aa4}. It is then easy to see by inspection that different single-qudit errors map each of the codewords to distinct, orthogonal states. Specifically, we show that the codewords in Eq.~\eqref{eq:4qdt_code} satisfy,
\begin{eqnarray}
\bra{m_L} A^{\dagger}_{x000} A_{000y} \ket{n_L} &\propto& \delta_{m,n} \delta_{x,0} \delta_{y,0} \nonumber \\
\bra{m_L} A^{\dagger}_{x000} A_{y000} \ket{n_L} &=& f_{x}(d)\cO(\gamma^{x})\delta_{m,n}\delta_{x,y} \nonumber \\
    \bra{m_L}A_{0000}^{\dag}A_{0000}\ket{n_L} &=& (1-2(d-1)\gamma + \cO(\gamma^{2}))\delta_{m,n} . \nonumber \\
    \label{eq:single_qudit}
\end{eqnarray}

Furthermore, we show by explicit calculation that the similar orthogonality conditions hold true for the two-qudit errors in the set $\cA_{\rm corr}$.

\begin{align}\label{eq:multi_error}
    \bra{m_L}A_{\alpha}^{\dag}A_{\beta}\ket{n_L} &= \frac{1}{d} \sum_{i=0}^{d-1} (i(i+m)_d \gamma^2 + \cO(\gamma^3))\delta_{m,n}\delta_{\alpha, \beta}\\
 &\quad \forall A_{\alpha},A_{\beta}\in \{A_{1010},A_{0101},A_{1001},A_{0110}\}.\nonumber
\end{align}
From Eqs.~\eqref{eq:single_qudit} and~\eqref{eq:multi_error}, we see that for every pair of errors $A_{\alpha}, A_{\beta} \in \cA_{\rm corr}$, the codewords of the $4$-qudit code satisfy, 
\[ \langle m_{L} | A^{\dagger}_{\alpha}A_{\beta}|n_{L}\rangle = \left( c_{\alpha,\beta} + g_{\alpha,\beta}(m,d)\right)\delta_{m,n}, \]
where, $c_{\alpha,\beta}$ are scalars of order $\cO(\gamma)$ that are independent of the choice of codewords $m,n$, and, $g_{\alpha,\beta}(m,d)$ are scalar functions of order $\cO(\gamma^{2})$ or higher that depend on the codeword index $m$ and the codespace dimension $d$.

We have thus shown that for the errors in the set $\cA_{\rm corr}$,  the $4$-qudit code defined in Eq.~\eqref{eq:4qdt_code} satisfy the form of the approximate QEC conditions stated in Eq.~\eqref{eq:appx_cond}, where the perturbative terms are orthogonal in the codeword labels ($m,n$) and are of order $\gamma^{2}$ or higher.
\end{proof}

Having verified that the $4$-qudit code satisfies the approximate QEC conditions, we may draw upon the results in~\cite{Leung, cafaro} to conclude that  {the} codes can correct all the errors in $\cA_{\rm corr}$ with a fidelity loss of $\cO(\gamma^2)$. The leading order $\cO(\gamma^2)$ term in the fidelity loss is, in fact, caused by two two-qudit error operators -- $A_{1100}$ and $A_{0011}$ -- that lead to logical errors. In what follows, we will explicitly construct a syndrome-based recovery scheme that retrieves the encoded information from the corrupted state with a fidelity loss of $\cO(\gamma^{2})$.

\section{Syndrome-Based Recovery for approximate four-qudit code} \label{sec:stabilizer for the code}

In principle,  {the} $[4,1]_d$ code can correct for all single-qudit errors and a subset of two-qudit errors that occur with a probability of order up to $\cO(\gamma^2)$, as shown in Theorem~\ref{thm:qudit_qec}. There exist certain two-qudit errors that are beyond the correction capabilities of this code, resulting in a fidelity loss proportional to $\gamma^2$. Since our fidelity is limited by this factor, we may choose to disregard single-qudit errors of order $\gamma^{3}$ or higher, whose contribution to the fidelity loss is anyway of $\cO(\gamma^3)$. 

We therefore construct a recovery scheme that can correct errors corresponding to the Kraus operators in the set
\begin{eqnarray}
  \cA^{(1)} &=& \{A_{0000}, A_{1000}, A_{0100}, A_{0010}, A_{0001}, \nonumber \\
            && A_{2000}, A_{0200}, A_{0020}, A_{0002}\} . \label{eq:IV1}
              \end{eqnarray}
In addition to these single-qudit errors, we show that our recovery scheme can also correct some two-qudit errors given by the Kraus operators
\begin{equation}
    \cA^{(2)} = \{A_{1010}, A_{1001}, A_{0110}, A_{0101}\}. \label{eq:IV2}
\end{equation} 

We first identify a set of stabilizer generators for  {the} $4$-qudit code, using the $d$-dimensional Pauli operators $X_{d}$ and $Z_{d}$ defined in Eqs.~\eqref{eq:xd} and~\eqref{eq:Zd}. The codewords in Eq.~\eqref{eq:4qdt_code} are stabilized by the group generated by,

\begin{align}\label{eq:stab_gen}
    \cS= \langle X_d X_d X_d X_d, Z_d^sZ_d^{d-s}I_d I_d, I_dI_dZ_d^sZ_d^{d-s} \rangle ,
\end{align}
where $s$ can take an integer value in the interval $[1, d-1]$ and is co-prime with $d$. In other words, we can say that the value of $s$ is chosen such that, greatest common divisor $\text{GCD}(d,s) = 1$. 
 One can always choose $s=1$ without the loss of generality, and we will stick to it for the rest of this discussion.

To detect the errors due to the single-qudit errors in the set $\cA^{(1)}$ and the two-qudit errors in the set $\cA^{(2)}$, we initially measure the stabilizer generators $Z_dZ_d^{d-1}I_dI_d$ and $I_d I_d Z_d Z_d^{d-1}$, and, denote the corresponding measurement results as $p_1 $ and $p_2$ respectively. Depending on the location and type of error, measurement of the $Z_{d}$-type stabilizers can yield one of the $d$ roots of unity as outcome. We label these outcomes as integers between $0$ and $(d-1)$, so that the syndromes $p_{1}, p_{2} \in [0, d-1]$. The syndrome values associated with the errors in $\cA^{(1)}\cup\cA^{(2)}$ are displayed in Table \ref{tab:primary_syndrome}.

\begin{table}[t]
    \centering
   \begin{tabular}{ |p{1.3cm}||p{1.3cm}|p{1.3cm}|p{1.3cm}|p{1.5cm}| }
 \hline
 \multicolumn{5}{|c|}{ Primary Syndromes ($p_1,p_2$)} \\
 \hline
 Error & d=3 & d=4 & d=5  & d \\
 \hline
 $A_{1000}$     & 2,0    & 3,0     & 4,0   & $d-1,0$ \\
 $A_{0100}$     & 1,0    & 1,0     & 1,0   & $1,0$\\
 $A_{2000}$     & 1,0    & 2,0     & 3,0   & $d-2,0$\\
 $A_{0200}$     & 2,0    & 2,0     & 2,0   & $2,0$\\
 $A_{0010}$     & 0,2    & 0,3     & 0,4   & $0,d-1$\\
 $A_{0001}$     & 0,1    & 0,1     & 0,1   & $0,1$\\
 $A_{0020}$     & 0,1    & 0,2     & 0,3   & $0,d-2$\\
 $A_{0002}$     & 0,2    & 0,2     & 0,2   & $0,2$\\
 $A_{1010}$     & 2,2    & 3,3     & 4,4   & $d-1,d-1$\\
 $A_{1001}$     & 2,1    & 3,1     & 4,1   & $d-1,1$\\
 $A_{0101}$     & 1,1    & 1,1     & 1,1   & $1,1$ \\
 $A_{0110}$     & 1,2    & 1,3     & 1,4   & $1,d-1$\\
 \hline
\end{tabular}
    \caption{Primary syndrome string for different errors occurring with probability up to $\cO(\gamma^2)$}.
    \label{tab:primary_syndrome}
\end{table}

 We observe that syndromes $p_1, p_{2}$ associated with different errors in $\cA^{(1)}\cup\cA^{(2)}$ are unique whenever $d \geq 5$. However, for $d=3,4$, there are pairs of errors that have the same $(p_{1}, p_{2})$ values.  
 Measuring the two stabilizer generators $Z_dZ_d^{d-1}I_dI_d$, $I_dI_dZ_dZ_d^{d-1}$ is thus not adequate to uniquely identify the location and type of error. We will therefore refer to $(p_{1}, p_{2})$ as \emph{primary} syndromes 
 and perform additional measurements in the case of the qutrit ($d=3$) code and ququad ($d=4$) code, to obtain \emph{secondary} syndromes values denoted as $s_1$ and $s_2$.
 
We explicitly discuss the syndrome-based recovery procedure for the qutrit case in the following section. The explicit recovery for the ququad code is discussed in Appendix \ref{appendix:a1}.

Interestingly, like the qutrit and ququad case, the $4$-qubit code also requires a two-step syndrome extraction procedure, with both primary and secondary syndromes required for unambiguous detection of single-damping errors. Indeed, the syndrome-based recovery for the qubit code is well-established~\cite{fletcher_channel}, and a fault-tolerant implementation of the same was given in \cite{Aj}. However, the two-step syndrome-based recovery for the $4$-qubit code cannot correct the two-qubit damping errors in the set $\cA^{(2)}$, as these errors will always annihilate the $\ket{1_L}$ state, making them non-recoverable. However, these errors can still be detected unambiguously by measuring the stabilizers $\{ZZII, IIZZ\}$.

\section{\label{sec:qutrit}The $4$-qutrit code}
 The codewords of the $4$-qutrit code are obtained by setting $d=3$ in Eq.~\eqref{eq:4qdt_code}.
\begin{align}
    \ket{0_L}&= \frac{1}{\sqrt{3}}\, [\ket{0000} + \ket{1111} +\ket{2222}] \\
    \ket{1_L}&= \frac{1}{\sqrt{3}}\, [\ket{0011} + \ket{1122} +\ket{2200}] \\
    \ket{2_L}&= \frac{1}{\sqrt{3}}\, [\ket{1100} + \ket{2211} +\ket{0022}]
\end{align}

As already shown in Theorem~\ref{thm:qudit_qec}, the code resulting from the span of these codewords satisfies the approximate QEC conditions up to $\cO(\gamma^{2})$. Specifically, the set of single-damping and two-damping errors in the sets $\cA^{(1)}$ and $\cA^{(2)}$ defined in Sec.~\ref{sec:stabilizer for the code}, map the codewords into mutually orthogonal states. Furthermore, for each of the errors $A_{\alpha} \in \cA^{(1)}\cup\cA^{(2)}$, the matrix elements $\langle i_{L} |A_{\alpha}^{\dagger}A_{\alpha}|i_{L}\rangle$ that feature in the QEC conditions are independent of the codeword index $i$ up to $\cO(\gamma)$, making this an approximate quantum error correcting code. 
 
\begin{table}[t!]
    \centering
   \begin{tabular}{ |p{1.75cm}||p{2.1cm}|p{2.1cm}| }
 \hline
 %\multicolumn{}{}{}
\multicolumn{3}{|c|}{Table of syndromes for Qutrits }\\
\hline
 
 Errors & $p_1, p_2$ & $s_1, s_2$ \\
 \hline
 $A_{1000}$   & 2,0    & 1,0  \, {\rm or}~ 1,2  \\
 $A_{0100}$   & 1,0    & 1,0  \, {\rm or}~ 1,2   \\
 $A_{2000}$   & 1,0    & 0,0  \, {\rm or}~ 0,2   \\
 $A_{0200}$   & 2,0    & 0,0  \, {\rm or}~ 0,2   \\
 $A_{0010}$   & 0,2    & 0,1  \, {\rm or}~ 2,1   \\
 $A_{0001}$   & 0,1    & 0,1  \, {\rm or}~ 2,1  \\
 $A_{0020}$   & 0,1    & 0,0  \, {\rm or}~ 2,0   \\
 $A_{0002}$   & 0,2    & 0,0  \, {\rm or}~ 2,0  \\
 $A_{1010}$   & 2,2    & $\times$       \\
 $A_{1001}$   & 2,1    & $\times$       \\
 $A_{0101}$   & 1,1    & $\times$       \\
 $A_{0110}$   & 1,2    & $\times$       \\
 \hline
\end{tabular}
    \caption{Primary and secondary syndrome strings for errors with probability up to $\cO(\gamma^2)$ for the qutrit code.}
    \label{t2}
\end{table}

\subsection{Syndrome-based Recovery for the $4$-qutrit Code}
Looking at the primary syndrome values in Table~\ref{tab:primary_syndrome}, we observe that there are four specific situations in which the primary syndromes are unable to definitively distinguish the errors in the qutrit case, meaning that two different errors have the same primary syndromes ($p_1$ and $p_2$). If we consider the first two qutrits as one block and the last two qutrits as another block, the degeneracy in the primary syndromes arises due to a single damping error in one qutrit and a double-damping error in the other qutrit of the same block.

To distinguish between errors with overlapping values of $(p_{1}, p_{2})$, we measure two additional operators, namely $W_3 W_3 II$ and $IIW_3 W_3$ to obtain the so-called \emph{secondary} syndromes. Here,  $W_3$ is a single-qutrit operator defined as
\begin{align}
    W_3=\begin{bmatrix}
    1 &  0 & 0 \\
    0 & -1 & 0 \\ 
    0 &  0 & 1 \end{bmatrix},\label{zbar}
\end{align}
which leaves the quantum states $\ket{0}$ and $\ket{2}$ unchanged, but adds a global phase of $\pi$ to the state $\ket{1}$.
The choice of the operator $W_3$ to determine the secondary syndromes is motivated by its commutation properties shown below. 
\begin{align}\label{eq:anti_comm}
     A_k W_3 = 
\begin{cases}
    -W_3 A_{k},& \text{when } k = 1\\
    W_3 A_{k}, & \text{when } k = 0,2
\end{cases}
\end{align}
\begin{widetext}
\begin{center}
\begin{figure}[t]
\includegraphics[width = 0.9 \textwidth]{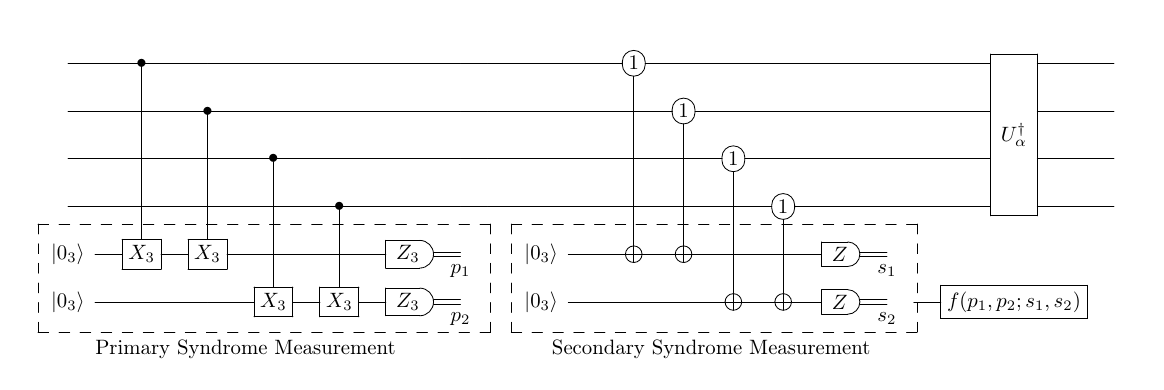}
\caption{Circuit diagram for the qutrit syndrome-based recovery. $p_1, p_2$ are the primary syndromes and $s_1,s_2$ are the secondary syndromes. Here $\ket{0_3}$ is the qutrit ground state. The controlled-$X_3$ gate is defined in Eq.\eqref{eq:cnotd}. The controlled gate denoted by the circle with a ``$1$'' in the secondary-syndrome measurement is the $C_3$ gate defined in Eq.~\eqref{cx3}. The function $f$ takes the syndromes $(p1, p2 ; s1, s2 )$ as input, identifies the error, and applies a suitable unitary operator $U_{\alpha}$ [as defined in Eq.\eqref{eq:uni_recov}] to recover the state.}

\label{fig:syn_circ}
\end{figure}
\end{center}
\end{widetext}

The controlled operation required for the secondary-syndrome measurements for $d=3$ is thus represented by

\begin{align}\label{cx3}
     C_3=\begin{bmatrix}
    \mathbf{I}_3 &  0 & 0 \\
    0 & \mathbf{X}_3 & 0 \\ 
    0 &  0 & \mathbf{I}_3 \end{bmatrix}.
\end{align}

The operator $C_3$ flips the target ancilla qutrit when the control qutrit is in the state $|1\rangle$. Thus, if adjacent data qutrits
are in the state $|1\rangle$,  the syndrome qutrits, which are initialized
to $|0\rangle
$ get mapped to the $|2\rangle$ state.

From Table \ref{t2}, we see that the secondary syndromes, in conjunction with the primary syndromes, can accurately identify the errors. The complete quantum circuit for extracting the primary and secondary syndromes for the $4$-qutrit code is illustrated in Fig.~\ref{fig:syn_circ}.

We note here that the secondary-syndrome measurements
of $W_3 W_3 II$ and $IIW_3 W_3$ can be performed via a controlled operation on ancillary controlled on the data qutrits.
As qutrit systems are more prone to decay compared to qubit
systems, qubit ancilla may provide better performance than
qutrit ancilla. Hence, one can try to use qubit-qutrit hybrid
setups \cite{bakkT,Yang,sdogra} for performing the recovery of the four-qutrit
code.

After identifying the errors using the primary and secondary syndromes, we map back the noisy state to the codespace by using a unitary recovery operator that is obtained as follows. Let $P_{\cC}$ denote the projector onto the $4$-qutrit subspace and $P_{\alpha}$ denote the projector associated with the syndrome measurement corresponding to the error $A_{\alpha} \in \cA^{(1)}\cup\cA^{(2)}$. Suppose the error $A_{\alpha}$ is detected, we construct the unitary operator $U_{\alpha}$ via polar decomposition of the non-unitary operator $A_{\alpha}P_{\cC}$, as,
\begin{align}
    A_{\alpha}P_{\cC} = U_{\alpha}\sqrt{P_{\cC}A_{\alpha}^{\dagger}A_{\alpha}P_{\cC}}. \label{eq:uni_recov}
\end{align}
The noisy state can be recovered by applying the inverse of the unitary operator $U_{\alpha}$ followed by the projector on the codespace $P_{\cC}$. The recovery operation associated with error $A_{\alpha}$ can thus be expressed as
\begin{equation}
    R_{\alpha} = P_{\cC} U_{\alpha}^{\dagger}P_{\alpha} \label{eq:Rk}.
\end{equation}
Note that $\sum_{\alpha} R_{\alpha}^{\dagger}R_{\alpha} \leq I $, since the projectors $\{P_{\alpha}\}$ are mutually orthogonal. Thus, the the set of Kraus operators $\{R_{\alpha}\}$, with the additional operator $\sqrt{I - \sum_{\alpha}R^{\dagger}_{\alpha}R_{\alpha}}$ constitutes a CPTP noise-adapted recovery channel. 

It is worth noting here that our syndrome-based recovery is the same as the noise-adapted recovery maps due to Leung {\emph et al.}~\cite{Leung} and Cafaro \emph{et al.}~\cite{cafaro} defined in Sec.~\ref{sec:NA_rec}. Since the $4$-qutrit code satisfies the orthogonal form of the AQEC conditions in Eq.~\eqref{eq:aqec_orth},  {it} satisfies the conditions of Lemma~\ref{lem: equiv}, making these two noise-adapted recovery maps identical in our case.

\begin{figure}[t]
  \centering
  \includegraphics[scale=0.64]{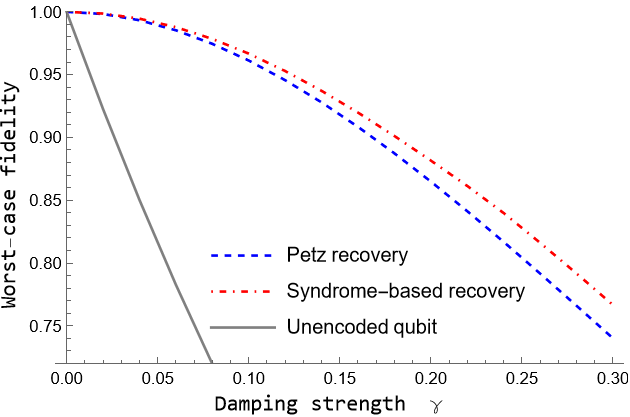}
  \caption{Worst-case fidelity vs the damping strength $\gamma$ for the $4$-qutrit code with the Petz recovery vs syndrome-based recovery.}
  \label{fig:qutrit_petz_leung}
\end{figure}

\subsection{Fidelity-loss and Comparison with the Petz Recovery}

We now compare the performance of our syndrome-based noise-adapted recovery with the universal near-optimal recovery map, namely, the Petz map defined via the Kraus operators in Eq.~\eqref{eq:petz}. 
We use the worst-case fidelity defined in \eqref{eq: wcf} to benchmark the performance of  {the} $4$-qutrit code in conjunction with each of these two recovery schemes. 

As shown in Fig.~\eqref{fig:qutrit_petz_leung}, the syndrome-based recovery performs better than the Petz recovery for the $4$-qutrit code subject to amplitude-damping noise for a broad range of damping strength $\gamma$. This observation is important as the Petz recovery is known to be near-optimal for the worst-case fidelity for any choice of code and noise~\cite{prabha}.
This behaviour also contrasts with the behaviour of the $4$-qubit code under amplitude-damping noise, where the Petz recovery outperforms the syndrome-based recovery. This deviation from the qubit case is probably due to the fact that the class of $4$-qudit code can correct for more errors when $d \geq 3$. Specifically, while for $d \geq 3$, the code can correct for the two-qudit errors in the set $\cA^{(2)}$ (see Eq.~\eqref{eq:IV2}),  for $d=2$, the $4$-qubit code can only detect but not correct for these two-qubit errors.

To gain a more detailed understanding, we have plotted the fidelity obtained using the 
syndrome-based recovery and the Petz recovery for different logical states subject to amplitude-damping noise. For the qutrit system, we can write any state in the codespace as,
\begin{align}
    \ket{\psi_L}& = \cos{\theta_1} \cos{\theta_2} \ket{0_L} + e^{i \phi_1} \cos{\theta_1} \sin{\theta_2} \ket{1_L}\nonumber\\
    & \qquad + e^{i \phi_2} \sin{\theta_1} \ket{2_L}. \label{eq: qutrit}
\end{align}

Fig.~\ref{fig:qutrit_petz_leung_state} shows the fidelity between the initial state and the recovered state for different values of $\theta_1$ and $\theta_2$, while setting the relative phases $\phi_1$ and $\phi_2$ to be zero. We first observe that the fidelity for the syndrome-based recovery is less sensitive to the choice of state in comparison to the Petz recovery. It is also clear from Fig. \ref{fig:qutrit_petz_leung_state}, that for a certain range of states, the Petz map performs better than syndrome-based recovery. 

\begin{figure}[t]
  \centering
  \includegraphics[scale=.67]{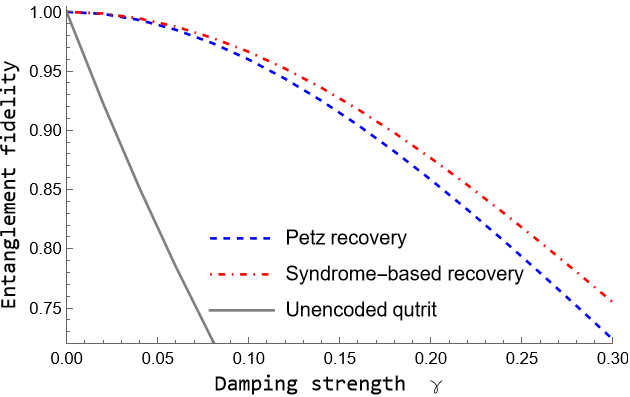}
  \caption{Entanglement fidelity vs the damping strength $\gamma$ for the $4$-qutrit code with Petz recovery vs syndrome-based recovery.}
  \label{fig:qutrit_ent_fid}
\end{figure}

Finally, we also estimate the entanglement fidelity for both recovery schemes. For the $4$-qutrit code, the Petz recovery yields an entanglement fidelity of 
\begin{equation}
    \cF_{Petz} = 1-4.52\gamma^2 - \cO(\gamma^3), \label{eq:ent_fidP}
\end{equation} whereas for the syndrome-based recovery, the entanglement fidelity is,
\begin{equation}
    \cF_{\rm syn} = 1-3.62\gamma^2 - \cO(\gamma^3). \label{eq:ent_fidS}
\end{equation}
Fig.~\ref{fig:qutrit_ent_fid} confirms that the syndrome-based recovery does achieve a higher entanglement fidelity than the Petz recovery for the $4$-qutrit code.

\subsection{Performance of noise-adapted qu$d$it code beyond $d=3$}
As noise-adapted QEC codes do not correct for arbitrary errors, it follows that these codes can, at times, beat the quantum  {Singleton} bound. Recall that the quantum {Singleton} bound gives us an upper bound on the number of qudits needed to protect a certain number of qudits from arbitrary errors. The formula that determines the {Singleton} bound for an $[[n, k, D]]_d$ code, where $d$ represents the local dimension of the quantum system and $D$ is the distance of the code which can correct up to $t = \lfloor \frac{D-1}{2} \rfloor$ qudits, is given by~\cite{singleton},
\begin{align}
     {n - k \leq 2(D-1)}.\label{eq:hamming}
\end{align}
We set $t$ and $k$ to be $1$ to know the minimum number of qudits required to protect a single logical qudit from any single-qudit physical error. 
It follows from the bound that we need a minimum of five qudits to protect one qudit. 
\begin{figure}[t]
  \centering
  \includegraphics[scale=0.52]{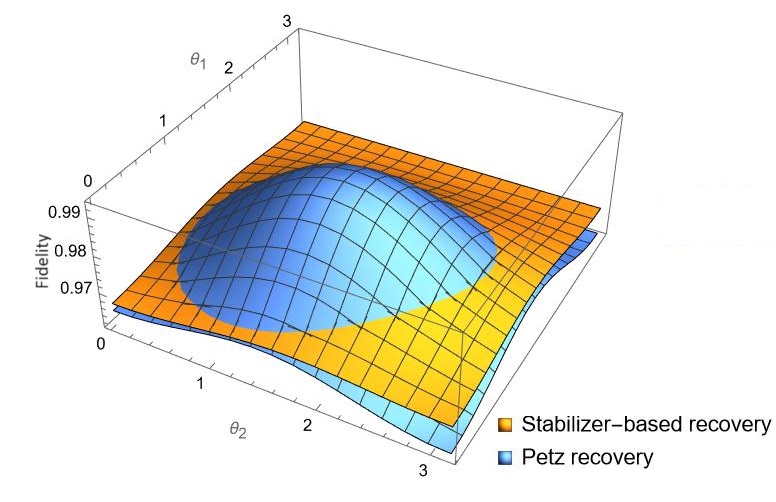}
  \caption{Fidelity of qutrit logical states parametrized by angles $\theta_1$ and $\theta_2$, with relative phases $\phi_1 = \phi_2 = 0$, for $\gamma=0.1$, using the Petz and syndrome-based recoveries. }
  \label{fig:qutrit_petz_leung_state}
\end{figure}

Hence,  {the class of $4$-qudit codes studied here} beats the quantum  {Singleton} bound for the  {qudit} case, similar to the $4$-qubit code~\cite{Leung}  {in the case of amplitude-damping noise}.

Increasing the local dimension $d$ of a quantum system increases the quantum information processing capacity exponentially. 
This naturally raises the question of how the performance of  {the class of codes defined in Eq.~\eqref{eq:4qdt_code}} scales with increasing the system dimension $d$. Based on our calculations in Appendix~\ref{sec:proof}, we can precisely calculate the extent to which  {this} class of $4$-qudit codes deviates from the Knill-Laflamme conditions for any $d$. As seen from Eqs.~\eqref{eq:no_damp}-~\eqref{eq:2_damp}, the leading order deviation is of $\cO(\gamma^{2})$ and the dimensional-dependent coefficient of this term grows as $d^{2}$ for all the correctable errors. The fidelity of the recovered state with respect to the encoded state is thus the form $\mathcal{F} \approx 1-\chi \gamma^2 + \mathcal{O}(\gamma^3)$,
where $\chi$ -- which captures the \emph{fidelity-loss} -- grows as $\cO(d^2)$. We also verify this numerically and plot the fidelity-loss coefficient $\chi$ as a function of $d$ in Fig. \ref{fig:fidlos}.

Thus, while  {the} class of $4$-qudit codes has the advantage that the number of stabilizer measurements remains fixed for $d \geq 5$, the performance of the code drops with increasing $d$, with the fidelity-loss scaling quadratically in $d$. 
\begin{figure}[t]
    \centering
    \includegraphics[scale=0.65]{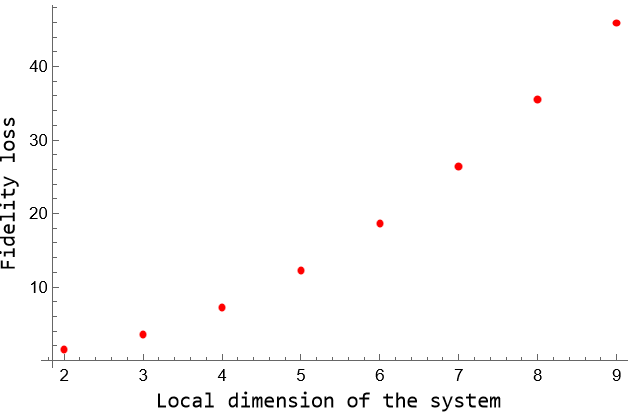}
    \caption{Fidelity loss ($\chi$) vs local dimension $d$ under the Petz recovery.}
    \label{fig:fidlos}
\end{figure}

\section{Generalization to a class of $[2M+2,M]_d$ qudit codes}\label{section:generalization}

Finally, we show that we can extend the $[4, 1]_d$ noise-adapted code to a family of $[(2M + 2), M]_d$ codes encoding $M$ qudits and capable of correcting AD noise to first order in noise strength. For $M = 1$, this family corresponds to the $[4, 1]_d$ noise-adapted qudit code discussed in the previous sections. A similar extension of the $[4, 1]_2$ qubit code to a family of $[(2M + 2), M]_2$ qubit codes for amplitude-damping noise has been described in \cite{fletcher_channel}.

The stabilizer generators for the $[(2M + 2), M]_d$ qudit codes are obtained as a straightforward extension of the stabilizer generators for the $[4, 1]_d$ code in Eq.~\eqref{eq:stab_gen} as
\begin{align}\label{eq:stab_M}
    \begin{split}
        \cS = \langle X_d&^{\otimes 2M+2}, Z_d^{d-s} Z_d^{s} I_d^{\otimes 2M}, I_d^{\otimes 2} Z_d^{d-s} Z_d^{s} I_d^{\otimes 2M-2},\\ \cdots,
    &I_d^{\otimes 2M-2} Z_d^{d-s} Z_d^{s} I_d^{\otimes 2}, I_d^{\otimes 2M+2} Z_d^{d-s} Z_d^{s} \rangle.
    \end{split}
\end{align}
Here, $X_d$ and $Z_d$ are the qudit Pauli operators defined in Eqs. \eqref{eq:xd} and \eqref{eq:Zd}. 
There is a single $X-$type stabilizer and $M+1$ Z-type stabilizers in the stabilizer generators, enabling it to encode $M$ qudits. 
The value of $s$ in Eq. \eqref{eq:stab_M} is chosen such that $s$ is co-prime with $d$.
The codewords for the qudit $[2M+2,M]$ code is given by
\begin{align}
    \begin{split}
        \ket{m_1,m_2, \cdots, m_M} = \frac{1}{\sqrt{d}} \sum_{i=0}^{d-1} \ket{i}_1 \ket{i}_2 \ket{i+m_1}_3 \ket{i+m_1}_4 \\  \ket{i+m_2}_5 \ket{i+m_2}_6 \cdots \ket{i+m_M}_{2M+1} \ket{i+m_M}_{2M+2}
    \end{split}  
\end{align}
As a concrete example, we present the codewords for the $M=2$ and $d=3$ case, corresponding to a $[6, 2]_3$ qutrit code, 
\begin{align}
    \begin{split}
        \ket{00}_L = \ket{000000} + \ket{111111} + \ket{222222} \\
        \ket{01}_L = \ket{000011} + \ket{111122} + \ket{222200} \\
        \ket{02}_L = \ket{000022} + \ket{111100} + \ket{222211} \\
        \ket{10}_L = \ket{001100} + \ket{112211} + \ket{220022} \\
        \ket{11}_L = \ket{001111} + \ket{112222} + \ket{220000} \\
        \ket{12}_L = \ket{001122} + \ket{112200} + \ket{220011} \\
        \ket{20}_L = \ket{002200} + \ket{110011} + \ket{221122} \\
        \ket{21}_L = \ket{002211} + \ket{110022} + \ket{221100} \\
        \ket{22}_L = \ket{002222} + \ket{110000} + \ket{221111}
    \end{split}
\end{align}
To correct all the single-qudit AD errors, we first need to perform a set of measurements leading to the secondary syndromes. 
The number of secondary measurements required to distinguish between the single-qudit errors is $M+1$. The structure of these secondary-syndrome measurement operators depends on the local dimension of the qudit. 
When $d=2$, the secondary syndrome measurement operators are $\{Z I^{\otimes 2M+1}, I^{\otimes 2} Z I^{\otimes 2M-1} , \cdots, I^{\otimes 2M+1} Z \}$. For $d=3$, the secondary syndrome measurement operators are $\{ W_3 W_3 I^{\otimes 2M}, I^{\otimes 2}W_3W_3I^{\otimes 2M-2}, \cdots, I^{\otimes 2M} W_3W_3\}$, where, $W_3$ is given in Eq. \eqref{zbar}.
For $d=4$, the extra measurements are given by $\{W_4 I^{\otimes 2M+1}, I^{\otimes 2} W_4 I^{\otimes 2M-1} , \cdots, I^{\otimes 2M+1} W_4 \}$.
The operator $W_4$ has the matrix representation,
\begin{align}
    W_4=\begin{bmatrix}
    1 &  0 & 0 & 0 \\
    0 &  1 & 0 & 0 \\ 
    0 &  0 &-1 & 0 \\
    0 &  0 & 0 &-1 \end{bmatrix}\label{zhat}
\end{align}
We can distinguish and correct all the first-order and some second-order errors using the primary syndromes obtained from measuring the $Z$-type stabilizers and the secondary syndromes from the additional measurements.

\section{Summary and Future Directions}\label{sec:conclusion}

In this work, we describe a four-qudit quantum error correcting code tailored to protect against amplitude-damping noise. We show that this code satisfies a special form of the approximate QEC conditions with the deviation term being of $\cO(\gamma^{2})$ where $\gamma$ is the probability of a single jump or single-damping event. 
Apart from the stabilizer generators, we identify a set of additional measurements that are required to realize a noise-adapted recovery scheme that corrects for all single-qudit and some two-qudit errors, thus achieving the desired order ($\cO(\gamma^{2})$) of fidelity.

 This code includes the well-known $4$-qubit amplitude-damping code~\cite{Leung} as a special case. Similar to the $4$-qubit case~\cite{fletcher_channel, Aj}, we see that the syndrome-based recovery for the $4$-qutrit and $4$-quqad codes also requires a two-step syndrome extraction procedure with two primary and two secondary syndrome measurements.

However, for $d\geq 5$, we do not need any secondary syndromes to distinguish the single-qudit errors of damping strength up to $\cO(\gamma^2)$, so the number of syndrome measurements does not grow with the local dimension of the code.

To benchmark the performance of the syndrome-based recovery, we consider the specific case of the $4$-qutrit code and calculate both the wort-case and entanglement fidelities for this case. We then compare these fidelities with those achieved by the Petz recovery map for the same code and noise channel. Interestingly, in a departure from the qubit case, we find that the syndrome-based recovery method yields higher fidelities than the Petz recovery. We believe that this trend should hold true for higher dimensions as well, since for all $d\geq 3$, this class of codes corrects some subset of two-qudit errors in addition to the single-qudit errors.

Finally, we have generalized the $4$-qudit code to a family of $[(2M+2), M]$ qudit codes, which are capable of correcting all single-qudit and a few two-qudit AD errors with similar secondary syndrome measurements as the four-qudit code.

Going forward, this $4$-qudit code and the associated syndrome-based recovery presented here could provide a promising avenue for exploring fault-tolerant quantum computation using qudits that predominantly undergo amplitude-damping noise. We know from the work of Akshaya \emph{et al..} \cite{Aj} that the $4$-qubit code can be used to achieve fault tolerance against amplitude-damping noise, with a finite pseudothreshold in the noise strength $\gamma$. It would, therefore,
be interesting to study how such a pseudothreshold scales
with the local system dimension $d$, using the codes and the
syndrome-based recovery discussed for the $[4, 1]_d$ code. We
can also study if a similar fault-tolerance threshold exists for
the $[2M + 2, M]_d$ code and how such a threshold may scale
with $M$ and $d$.

\begin{acknowledgments}
We would like to acknowledge the grant from the Mphasis F1 Foundation to CQuICC. We are thankful to Akshaya Jayashankar for many insightful discussions and Shrikant Utagi and Vismay Joshi for helpful comments. We are also grateful to the anonymous referees for their valuable comments.   

\end{acknowledgments}

\bibliography{apssamp}

\providecommand{\noopsort}[1]{}\providecommand{\singleletter}[1]{#1}%
\begin{thebibliography}{10}

\bibitem{qudits_comp1}
Yuchen Wang, Zixuan Hu, Barry~C Sanders, and Sabre Kais.
\newblock Qudits and high-dimensional quantum computing.
\newblock {\em Frontiers in Physics}, 8:589504, 2020.

\bibitem{qudit_comp2}
Yulin Chi, Jieshan Huang, Zhanchuan Zhang, Jun Mao, Zinan Zhou, Xiaojiong Chen, Chonghao Zhai, Jueming Bao, Tianxiang Dai, Huihong Yuan, et~al.
\newblock A programmable qudit-based quantum processor.
\newblock {\em Nature communications}, 13(1):1166, 2022.

\bibitem{qudit_comp3}
MingXing Luo and XiaoJun Wang.
\newblock Universal quantum computation with qudits.
\newblock {\em Science China Physics, Mechanics \& Astronomy}, 57:1712--1717, 2014.

\bibitem{qudit_jaynes_cummings}
Brian Mischuck and Klaus M{\o}lmer.
\newblock Qudit quantum computation in the jaynes-cummings model.
\newblock {\em Physical Review A}, 87(2):022341, 2013.

\bibitem{qudit1}
S.~S. Ivanov, H.~S. Tonchev, and N.~V. Vitanov.
\newblock Time-efficient implementation of quantum search with qudits.
\newblock {\em Phys. Rev. A}, 85:062321, Jun 2012.

\bibitem{qudit2}
Anna~S Nikolaeva, Evgeniy~O Kiktenko, and Aleksey~K Fedorov.
\newblock Efficient realization of quantum algorithms with qudits.
\newblock {\em arXiv preprint arXiv:2111.04384}, 2021.

\bibitem{quditvqe}
Shuxiang Cao, Mustafa Bakr, Giulio Campanaro, Simone~D Fasciati, James Wills, Deep Lall, Boris Shteynas, Vivek Chidambaram, Ivan Rungger, and Peter Leek.
\newblock Emulating two qubits with a four-level transmon qudit for variational quantum algorithms.
\newblock {\em arXiv preprint arXiv:2303.04796}, 2023.

\bibitem{qutrit}
A.~S. {Nikolaeva}, E.~O. {Kiktenko}, and A.~K. {Fedorov}.
\newblock {Decomposing the generalized Toffoli gate with qutrits}.
\newblock {\em \pra}, 105(3):032621, March 2022.

\bibitem{terhal_qec}
Barbara~M Terhal.
\newblock Quantum error correction for quantum memories.
\newblock {\em Reviews of Modern Physics}, 87(2):307, 2015.

\bibitem{fowler2012}
Austin~G Fowler, Matteo Mariantoni, John~M Martinis, and Andrew~N Cleland.
\newblock Surface codes: Towards practical large-scale quantum computation.
\newblock {\em Physical Review A}, 86(3):032324, 2012.

\bibitem{gottesman}
Daniel Gottesman.
\newblock {\em Stabilizer codes and quantum error correction}.
\newblock California Institute of Technology, 1997.

\bibitem{qudit_hamming}
Sixia Yu, C.~H. Lai, and C.~H. Oh.
\newblock Strengthened quantum hamming bound, 2010.

\bibitem{singleton}
Markus Grassl, Felix Huber, and Andreas Winter.
\newblock Entropic proofs of singleton bounds for quantum error-correcting codes.
\newblock {\em IEEE Transactions on Information Theory}, 68(6):3942--3950, 2022.

\bibitem{sourenoise}
Miha Papič, Adrian Auer, and Inés de~Vega.
\newblock Fast estimation of physical error contributions of quantum gates, 2023.

\bibitem{Leung}
Debbie~W. Leung, M.~A. Nielsen, Isaac~L. Chuang, and Yoshihisa Yamamoto.
\newblock Approximate quantum error correction can lead to better codes.
\newblock {\em Phys. Rev. A}, 56:2567--2573, Oct 1997.

\bibitem{barnum2002}
Howard Barnum and Emanuel Knill.
\newblock Reversing quantum dynamics with near-optimal quantum and classical fidelity.
\newblock {\em Journal of Mathematical Physics}, 43(5):2097--2106, 2002.

\bibitem{prabha}
Hui~Khoon Ng and Prabha Mandayam.
\newblock Simple approach to approximate quantum error correction based on the transpose channel.
\newblock {\em Phys. Rev. A}, 81:062342, Jun 2010.

\bibitem{fletcher_channel}
Andrew~S Fletcher, Peter~W Shor, and Moe~Z Win.
\newblock Channel-adapted quantum error correction for the amplitude damping channel.
\newblock {\em IEEE Transactions on Information Theory}, 54(12):5705--5718, 2008.

\bibitem{preskill2018}
John Preskill.
\newblock Quantum computing in the nisq era and beyond.
\newblock {\em Quantum}, 2:79, 2018.

\bibitem{preskill_biased}
Panos Aliferis and John Preskill.
\newblock Fault-tolerant quantum computation against biased noise.
\newblock {\em Physical Review A}, 78(5):052331, 2008.

\bibitem{bias_camp}
Joschka Roffe, Lawrence~Z. Cohen, Armanda~O. Quintavalle, Daryus Chandra, and Earl~T. Campbell.
\newblock Bias-tailored quantum {LDPC} codes.
\newblock {\em {Quantum}}, 7:1005, May 2023.

\bibitem{liang_jiang_biased}
Qian Xu, Nam Mannucci, Alireza Seif, Aleksander Kubica, Steven~T. Flammia, and Liang Jiang.
\newblock Tailored xzzx codes for biased noise.
\newblock {\em Phys. Rev. Res.}, 5:013035, Jan 2023.

\bibitem{s_puri_bias}
Shruti Puri, Lucas St-Jean, Jonathan~A. Gross, Alexander Grimm, Nicholas~E. Frattini, Pavithran~S. Iyer, Anirudh Krishna, Steven Touzard, Liang Jiang, Alexandre Blais, Steven~T. Flammia, and S.~M. Girvin.
\newblock Bias-preserving gates with stabilized cat qubits.
\newblock {\em Science Advances}, 6(34):eaay5901, 2020.

\bibitem{jayashankar2020finding}
Akshaya Jayashankar, Anjala~M Babu, Hui~Khoon Ng, and Prabha Mandayam.
\newblock Finding good quantum codes using the cartan form.
\newblock {\em Physical Review A}, 101(4):042307, 2020.

\bibitem{Addnoise}
Isaac~L. Chuang, Debbie~W. Leung, and Yoshihisa Yamamoto.
\newblock Bosonic quantum codes for amplitude damping.
\newblock {\em Phys. Rev. A}, 56:1114--1125, Aug 1997.

\bibitem{wasilewski}
Wojciech Wasilewski and Konrad Banaszek.
\newblock Protecting an optical qubit against photon loss.
\newblock {\em Physical Review A}, 75(4):042316, 2007.

\bibitem{permutation_AD}
Yingkai Ouyang and Rui Chao.
\newblock Permutation-invariant constant-excitation quantum codes for amplitude damping.
\newblock {\em IEEE Transactions on Information Theory}, 66(5):2921--2933, 2020.

\bibitem{bin_code_vvalbert}
Marios~H. Michael, Matti Silveri, R.~T. Brierley, Victor~V. Albert, Juha Salmilehto, Liang Jiang, and S.~M. Girvin.
\newblock New class of quantum error-correcting codes for a bosonic mode.
\newblock {\em Phys. Rev. X}, 6:031006, Jul 2016.

\bibitem{bosonics_vva}
Victor~V. Albert.
\newblock Bosonic coding: introduction and use cases, 2022.

\bibitem{cat_1}
Zaki Leghtas, Gerhard Kirchmair, Brian Vlastakis, Robert~J. Schoelkopf, Michel~H. Devoret, and Mazyar Mirrahimi.
\newblock Hardware-efficient autonomous quantum memory protection.
\newblock {\em Phys. Rev. Lett.}, 111:120501, Sep 2013.

\bibitem{cat_2}
P.~T. Cochrane, G.~J. Milburn, and W.~J. Munro.
\newblock Macroscopically distinct quantum-superposition states as a bosonic code for amplitude damping.
\newblock {\em Phys. Rev. A}, 59:2631--2634, Apr 1999.

\bibitem{cat_3}
Mazyar Mirrahimi, Zaki Leghtas, Victor~V Albert, Steven Touzard, Robert~J Schoelkopf, Liang Jiang, and Michel~H Devoret.
\newblock Dynamically protected cat-qubits: a new paradigm for universal quantum computation.
\newblock {\em New Journal of Physics}, 16(4):045014, apr 2014.

\bibitem{s_puri}
Andrew~S. Darmawan, Benjamin~J. Brown, Arne~L. Grimsmo, David~K. Tuckett, and Shruti Puri.
\newblock Practical quantum error correction with the xzzx code and kerr-cat qubits.
\newblock {\em PRX Quantum}, 2:030345, Sep 2021.

\bibitem{s_puri_1}
Alexander Grimm, Nicholas~E Frattini, Shruti Puri, Shantanu~O Mundhada, Steven Touzard, Mazyar Mirrahimi, Steven~M Girvin, Shyam Shankar, and Michel~H Devoret.
\newblock The kerr-cat qubit: stabilization, readout, and gates.
\newblock {\em arXiv preprint arXiv:1907.12131}, 2019.

\bibitem{review_bose_code}
Weizhou Cai, Yuwei Ma, Weiting Wang, Chang-Ling Zou, and Luyan Sun.
\newblock Bosonic quantum error correction codes in superconducting quantum circuits.
\newblock {\em Fundamental Research}, 1(1):50--67, 2021.

\bibitem{mg}
Markus Grassl, Linghang Kong, Zhaohui Wei, Zhang-Qi Yin, and Bei Zeng.
\newblock Quantum error-correcting codes for qudit amplitude damping.
\newblock {\em IEEE Transactions on Information Theory}, 64(6):4674--4685, 2018.

\bibitem{Aj}
Akshaya Jayashankar, My~Duy~Hoang Long, Hui~Khoon Ng, and Prabha Mandayam.
\newblock Achieving fault tolerance against amplitude-damping noise.
\newblock {\em Phys. Rev. Res.}, 4:023034, Apr 2022.

\bibitem{junge2018}
Marius Junge, Renato Renner, David Sutter, Mark~M Wilde, and Andreas Winter.
\newblock Universal recovery maps and approximate sufficiency of quantum relative entropy.
\newblock In {\em Annales Henri Poincar{\'e}}, volume~19, pages 2955--2978. Springer, 2018.

\bibitem{KLCondition}
Emanuel Knill, Raymond Laflamme, and Lorenza Viola.
\newblock Theory of quantum error correction for general noise.
\newblock {\em Phys. Rev. Lett.}, 84:2525--2528, Mar 2000.

\bibitem{nielsen}
Michael~A. Nielsen and Isaac~L. Chuang.
\newblock {\em Quantum Computation and Quantum Information}.
\newblock Cambridge University Press, 2000.

\bibitem{beny}
C\'edric B\'eny and Ognyan Oreshkov.
\newblock General conditions for approximate quantum error correction and near-optimal recovery channels.
\newblock {\em Phys. Rev. Lett.}, 104:120501, Mar 2010.

\bibitem{cafaro}
Carlo Cafaro and Peter van Loock.
\newblock Approximate quantum error correction for generalized amplitude-damping errors.
\newblock {\em Phys. Rev. A}, 89:022316, Feb 2014.

\bibitem{dom_ampd}
Luca Chirolli and Guido Burkard.
\newblock Decoherence in solid-state qubits.
\newblock {\em Advances in Physics}, 57(3):225--285, 2008.

\bibitem{gilyen2022}
Andr{\'a}s Gily{\'e}n, Seth Lloyd, Iman Marvian, Yihui Quek, and Mark~M Wilde.
\newblock Quantum algorithm for petz recovery channels and pretty good measurements.
\newblock {\em Physical Review Letters}, 128(22):220502, 2022.

\bibitem{Biswas}
Debjyoti Biswas, Gaurav~M Vaidya, and Prabha Mandayam.
\newblock Noise-adapted recovery circuits for quantum error correction.
\newblock {\em arXiv preprint arXiv:2305.11093}, 2023.

\bibitem{group_code}
Philippe Faist, Sepehr Nezami, Victor~V. Albert, Grant Salton, Fernando Pastawski, Patrick Hayden, and John Preskill.
\newblock Continuous symmetries and approximate quantum error correction.
\newblock {\em Phys. Rev. X}, 10:041018, Oct 2020.

\bibitem{bakkT}
T.~B{\ae}kkegaard, L.~B. Kristensen, N.~J.~S. Loft, C.~K. Andersen, D.~Petrosyan, and N.~T. Zinner.
\newblock Realization of efficient quantum gates with a superconducting qubit-qutrit circuit.
\newblock {\em Scientific Reports}, 9(1):13389, Sep 2019.

\bibitem{Yang}
Chui-Ping Yang, Zhen-Fei Zheng, and Yu~Zhang.
\newblock Universal quantum gate with hybrid qubits in circuit quantum electrodynamics.
\newblock {\em Opt. Lett.}, 43(23):5765--5768, Dec 2018.

\bibitem{sdogra}
Shruti Dogra, Arvind Dorai, and Kavita Dorai.
\newblock Implementation of the quantum fourier transform on a hybrid qubit–qutrit nmr quantum emulator.
\newblock {\em International Journal of Quantum Information}, 13(07):1550059, 2015.

\bibitem{KITAEV20032}
A.Yu. Kitaev.
\newblock Fault-tolerant quantum computation by anyons.
\newblock {\em Annals of Physics}, 303(1):2--30, 2003.

\bibitem{Andersen2020}
Christian~Kraglund Andersen, Ants Remm, Stefania Lazar, Sebastian Krinner, Nathan Lacroix, Graham~J. Norris, Mihai Gabureac, Christopher Eichler, and Andreas Wallraff.
\newblock Repeated quantum error detection in a surface code.
\newblock {\em Nature Physics}, 16(8):875–880, June 2020.

\end{thebibliography}

\appendix
\begin{widetext}

\section{Proof of Lemma \ref{lem: equiv}}\label{sec:app_leung}
Given a code with codewords $\{|i_{L}\rangle, i=1,\ldots,d\}$, and a set of Kraus operators $\{E_{k}\}$ that satisfy Eqs.~\eqref{eq:aqec_leung)} and~\eqref{eq:aqec_orth}, we first put the two conditions together to get the following form of the AQEC conditions 
\begin{equation}
    \langle i_L| E_{k}^{\dag} E_l |j_L \rangle = \left( c_{kk} + \langle i_{L} |B_{kk} |j_{L}\rangle \right) \delta_{kl}\delta_{ij}. \label{eq:app_aqec}
\end{equation}  
Polar decomposition of the Kraus operators on the codespace leads to unitaries $U_{k}$, given by
\begin{align}\label{eq:app_polar_decomp}
    E_{k} P &= U_{k}\sqrt{PE_k^{\dag}E_kP} = U_{k}P\sqrt{PE_k^{\dag}E_kP}
\end{align}
It follows from the orthogonality in the $(k,l)$ indices in Eq.~\eqref{eq:app_aqec} that the unitaries $\{U_{k}\}$ satisfy,
\begin{align*}
    PU_k^{\dag}U_lP & \propto \delta_{kl}P
\end{align*}
The recovery due to Leung \emph {et al.} is then defined via the operators $R^{(L)}_{k} = PU_{k}^{\dagger}$. This can be rewritten as,
%Following the prescription of Leung et al.., we consider a residue operator $\pi$ as
%\begin{align}
 %   \pi_k &= \sqrt{PE_k^{\dag}E_kP}-\sqrt{\lambda_kp_k}P
%\end{align}
%Therefore the polar decompostion in Eq.\eqref{eq:app_polar_decomp} takes the following following form,
%\begin{align}
%    E_kP = U _k (\pi_k+\sqrt{\lambda_k pk}I)P
%\end{align}

%Therefore the for an approximate code the Eq.\eqref{eq:ortho_knill} takes following form 
%\begin{align}\label{eq:ortho_leung}
%    PE_k^{\dag}E_lP &= (\sqrt{\lambda_kp_k}+\pi_k^{\dag})(\sqrt{\lambda_kp_k}+\pi_k)\delta_{kl}P
%\end{align}
 
%We see that if the operator $PE_k^{\dag}E_kP$ is diagonal in the $\{i_L\}$ basis, we can re-write the Eq.\eqref{eq:app_polar_decomp} as
\begin{align}
    R_{k}^{(L)} &= PU_{k}^{\dag} = \left( \sqrt{PE_k^{\dag}E_kP} \right)^{-1}PE_k^{\dag} \nonumber \\
     &=\left(\sum\limits_{i,j=0}^{d-1}\ket{i_L}\bra{j_L}\bra{i_L}E_k^{\dag}E_k\ket{j_L}\right)^{-\frac{1}{2}}P E_k^{\dag} .\label{eq:inter1}
     \end{align}
     Now, if we also invoke the orthogonality in the codeword indices $(i,j)$, we see that the Leung \emph{et al.} recovery operators can be written as,
     \begin{equation}
   R_{k}^{(L)} =\sum\limits_{i=0}^{d-1}\frac{\ket{i_L}\bra{i_L} E_k^{\dag}}{\sqrt{\bra{i_L}E_k^{\dag}E_k\ket{i_L}}} ,\label{eq:inter2} 
\end{equation}
where we have used $\bra{i_L}E_k^{\dag}E_k\ket{j_L} = \bra{i_L}E_k^{\dag}E_k\ket{i_L} \delta_{ij} $. Recalling the definition of the Cafaro \emph{et al.} recovery operators $R_{k}^{(C)}$ from Eq.~\eqref{eq:cafaro_kraus}, we see that the R.H.S. of Eq.~\eqref{eq:inter2} is indeed the same as $R_{k}^{(C)}$, showing that the two recovery maps are indeed identical. 

Conversely, for a set of codewords $\{|i_{L}\rangle\}$ and Kraus operators $\{E_{k}\}$ for which the two set of recovery operators $\{R_{k}^{(L)}\}$ and $\{R_{k}^{(C)}\}$ are identical, we must have,
\[R_{k}^{(L)} = PU_{k}^{\dagger} = \sum\limits_{i=0}^{d-1}\frac{\ket{i_L}\bra{i_L} E_k^{\dag}}{\sqrt{\bra{i_L}E_k^{\dag}E_k\ket{i_L}}} .\]
From the definition of the unitaries $\{U_{k}\}$ in Eq.~\eqref{eq:inter1}, we see that this is true only if $\bra{i_L}E_k^{\dag}E_k\ket{j_L} \propto \delta_{ij}$, thus proving Lemma \ref{lem: equiv}.

\section{ {Connection with distance-two qudit surface code}}\label{appendix:D}
Here, we point out an interesting connection between the $4$-qudit code defined in Eq.~\eqref{eq:4qdt_code} and a certain class of surface codes. The surface code is a variant of Kitaev's toric code~\cite{KITAEV20032} and belongs to the class of topological quantum codes defined on a lattice, where the qubits are placed on the edges of the lattice, and the stabilizer generators are defined using plaquette and vertex operators.
Surface codes \cite{fowler2012} are toric codes without periodic boundary conditions and have a similar stabilizer structure. 

It is known that the $4$-qubit code that corrects for single amplitude-damping errors can be mapped to one of the distance-$2$ qubit surface codes \cite{Andersen2020}. We now show that a similar connection exists between the $4$-qudit codes for amplitude-damping and one of the distance-$2$ qudit surface codes.

The distance-2 qudit surface code can be defined using four different sets of stabilizer generators, namely,

 \begin{align}
     \begin{split}
    &\mathcal{S}_a = \{X_dX_dX_dX_d, Z_dZ^{-1}I_dI_d, I_dI_dZ_dZ^{-1}\}, \qquad \mathcal{S}_b = \{X_dX^{-1}_dX_dX^{-1}_d, Z_dZ_dI_dI_d, I_dI_dZ_dZ_d\}, \\
     &\mathcal{S}_c = \{Z_dZ_dZ_dZ_d, X_dX_d^{-1}I_dI_d, I_dI_dX_dX^{-1}_d\}, \qquad \mathcal{S}_d = \{Z_dZ^{-1}_dZ_dZ^{-1}_d, X_dX_dI_dI_d, I_dI_dX_dX_d\},
     \end{split}
 \end{align}
 each leading to a different set of codewords. All four sets of stabilizers can detect any single-qudit Pauli error. However, when the noise is non-Pauli, the performance differs for different codes. We first note that the $4$-qudit codewords in Eq. \eqref{eq:4qdt_code} are stabilized by the generators of form $\mathcal{S}_a$.  Furthermore, we can explicitly write down the codewords stabilized by $\mathcal{S}_b$ as,
\begin{eqnarray}\label{eq:new_code}
 \ket{m_L} =   \frac{1}{\sqrt{d}}\sum_{i=0}^{d-1} \ket{i}_1 \ket{((d-1)i)_d}_2 \ket{(i+m)_d}_3 \ket{((d-1)(i+m))_d}_4 , \nonumber
\end{eqnarray}
for all $m\in[0,d-1]$. It can be easily checked that the codewords in Eq.~\eqref{eq:new_code} satisfy the Knill-Laflamme conditions in Eq.~\eqref{a1} up to the first order in damping strength $\gamma$, similar to the code in Eq.~\eqref{eq:4qdt_code}. Thus, we see that codes with stabilizer generators $\mathcal{S}_a$ and $\mathcal{S}_b$ are indeed amplitude-damping QEC codes, correcting up to first order in noise strength.

A similar exercise shows that the codewords corresponding to the stabilizers in the sets $\mathcal{S}_c$ and $\mathcal{S}_d$ do not satisfy the Knill-Laflamme conditions even up to the first order of damping strength for the qudit amplitude-damping channel. Thus, the distance-$2$ stabilizer codes associated with the generators in $\mathcal{S}_c$ and $\mathcal{S}_d$ do not correspond to noise-adapted codes for amplitude-damping noise.   
Finally, one important point to note is that even for the codewords associated with the sets $\mathcal{S}_a$ and $\mathcal{S}_b$, the syndromes measured by these stabilizer generators cannot distinguish all the correctable amplitude-damping errors. Hence, additional measurements are required to determine the exact error that occurred, as described in Sec.~\ref{sec:stabilizer for the code} of the main text.

\section{Proof of Theorem \ref{thm:qudit_qec}}{\label{sec:proof}}
The errors on the four-qudit system have the form $A_{ijkl} = A_i \otimes A_j \otimes A_k \otimes A_l$ where the error $A_{ijkl}$ occurs with probability $\cO(\gamma^{i+j+k+l})$. Consider the action of the single-qudit errors in the set $\cA_{\rm corr}$ on the codewords defined in Eq.~\eqref{eq:4qdt_code}. 
For all $m,x \in [0,d-1]$, we note that,
\begin{align}
    A_{x000} \ket{m_L} &= \frac{1}{\sqrt{d}} \sum_{i=x}^{d-1} \sqrt{\binom{i}{x}} (1-\gamma)^{i + (i+m)_d} \left(\frac{\gamma}{1-\gamma}\right)^{\frac{x}{2}}   \ket{i-x}_1 \ket{i}_2 \ket{(i+m)_d}_3 \ket{(i+m)_d}_4 \label{eq:3aa1}\\
    A_{0x00} \ket{m_L} &= \frac{1}{\sqrt{d}} \sum_{i=x}^{d-1} \sqrt{\binom{i}{x}} (1-\gamma)^{i + (i+m)_d} \left(\frac{\gamma}{1-\gamma}\right)^{\frac{x}{2}}   \ket{i}_1 \ket{i-x}_2 \ket{(i+m)_d}_3 \ket{(i+m)_d}_4 \label{eq:3aa2}\\
    A_{00x0} \ket{m_L} &= \frac{1}{\sqrt{d}} \sum_{(i+m)_d=x}^{d-1} \sqrt{\binom{(i+m)_d}{x}} (1-\gamma)^{i + (i+m)_d} \left(\frac{\gamma}{1-\gamma}\right)^{\frac{x}{2}}  \ket{i}_1 \ket{i}_2 \ket{(i+m)_d-x}_3 \ket{(i+m)_d}_4 \label{eq:3aa3}\\
    A_{000x} \ket{m_L} &= \frac{1}{\sqrt{d}} \sum_{(i+m)_d=x}^{d-1} \sqrt{\binom{(i+m)_d}{x}} (1-\gamma)^{i + (i+m)_d}  \left(\frac{\gamma}{1-\gamma}\right)^{\frac{x}{2}} \ket{i}_1 \ket{i}_2 \ket{(i+m)_d}_3 \ket{(i+m)_d-x}_4. \label{eq:3aa4}
\end{align}

We can now prove that  {the class of $4$-qudit codes studied here} satisfy the Knill-Laflamme condition for error correction up to $\cO({\gamma})$ by taking the inner products of the states in Eqs. \eqref{eq:3aa1} through~\eqref{eq:3aa4}. 

We first consider a pair of single-qudit damping errors on different qudit locations.  
\begin{align}
    &\bra{m_L} A^{\dagger}_{x000} A_{000y} \ket{n_L}\nonumber\\
    &= \frac{1}{d}\sum\limits_{i=x,(j+n)_d=y}^{d-1} \sqrt{{i \choose x}{(j+n)_d \choose y}} (1-\gamma)^{i+j+(i+m)_d+(j+n)_d}\left(\frac{\gamma}{1-\gamma}\right)^{\frac{x+y}{2}} \nonumber \delta_{i-x,j} \delta_{i,j} \delta_{(i+m)_d,(j+n)_d} \delta_{(i+m)_d,(j+n)_d-y} \nonumber\\
    &= \frac{1}{d} \sum\limits_{i=x}^{d-1} \sqrt{{i \choose x}{(i+n)_d \choose y}} (1-\gamma)^{2i+2(i+m)_d} \left(\frac{\gamma}{1-\gamma}\right)^{\frac{x+y}{2}} \delta_{i-x,i} \delta_{(i+m)_d,(i+n)_d} \delta_{(i+m)_d,(i+n)_d-y} \label{eq:A2}
\end{align}
As $\delta_{i-x,i} \delta_{(i+m)_d,(i+n)_d} \delta_{(i+m)_d,(i+n)_d-y} = \delta_{m,n} \delta_{x,0} \delta_{y,0}$, the R.H.S of the Eq.~\eqref{eq:A2} is non-zero if and only if $m=n$ and $x=y=0$. 
Therefore, single-qudit errors occurring in different locations map the codewords in Eq.~\eqref{eq:4qdt_code} to mutually orthogonal states.  

Next, we consider pairs of single-qudit errors occurring in the same location.
\begin{align}\label{eq:app1}
    \bra{m_L}A_{x000}^{\dag}A_{y000}\ket{n_L}&= \frac{1}{d}\sum_{i=x,j=y}^{d-1} \sqrt{\binom{i}{x} \binom{j}{y}} (1-\gamma)^{i+j+(i+m)_d+(j+n)_d}  \left(\frac{\gamma}{1-\gamma} \right)^{\frac{x+y}{2}} \delta_{i-x,j-y} \delta_{i,j} \delta_{(i+m)_d,(j+n)_d} \nonumber\\
    &= \frac{1}{d}\sum_{i=x}^{d-1} \sqrt{\binom{i}{x} \binom{i}{y}} (1-\gamma)^{2i+(i+m)_d+(i+n)_d}  \left(\frac{\gamma}{1-\gamma} \right)^{\frac{x+y}{2}} \delta_{i-x,i-y} \delta_{(i+m)_d,(i+n)_d} \nonumber\\
    &= \frac{1}{d}\sum_{i=x}^{d-1} \binom{i}{x} (1-\gamma)^{2i+2(i+m)_d}  \left(\frac{\gamma}{1-\gamma} \right)^{x} \delta_{x,y} \delta_{m,n} \nonumber\\
    &= \frac{1}{d}\sum_{i=x}^{d-1} \gamma^x {i \choose x} \left(1-( p-x) \gamma+ (\frac{p(p-1)}{2}+\frac{x(x-1)}{2} - xp)\gamma^2 + \cO(\gamma^3) \right) \delta_{x,y} \delta_{m,n}. %\label{eq:1damp_same}
\end{align}
where, $p = 2i+2(i+m)_d$. The presence of $\delta_{x,y}$ and $\delta_{m,n}$ in Eq.~\eqref{eq:app1} confirms again that different single-qudit errors occurring at the same location also map the codewords in Eq.~\eqref{eq:4qdt_code} to orthogonal states.  

For the cases when $x=y$ in Eq.~\eqref{eq:app1}, we now calculate the coefficients of different orders of $\gamma$, as follows.
\begin{itemize}
\item The no-damping error operator $A_{0000}$ ($x=0$)
\begin{eqnarray}
    \bra{m_L}A_{0000}^{\dag}A_{0000}\ket{n_L} &= \frac{1}{d}\sum_{i=0}^{d-1} \left( 1- p \gamma  + \frac{p(p-1)}{2} \gamma^2 + \cO(\gamma^3) \right)  \delta_{m,n} \nonumber\\
    &= \left( 1 - 2(d-1) \gamma + \frac{1}{d} \sum_{i=0}^{d-1}  (i+(i+m)_d)(2i+2(i+m)_d-1) \gamma^2 + \cO(\gamma^3) \right)  \delta_{m,n}  \label{eq:no_damp}
\end{eqnarray}
\item The single-qudit, single-damping error operator $A_{1000}$ ($x=1$)
\begin{align}
    \bra{m_L}A_{1000}^{\dag}A_{10000}\ket{n_L} &= \frac{1}{d}\sum_{i=1}^{d-1} \gamma i \left(1-( p-1) \gamma+ (p +\frac{p(p-1)}{2})\gamma^2 + \cO(\gamma^3) \right) \delta_{m,n} \nonumber\\
    &= \left( \frac{d-1}{2} \gamma + \frac{1}{d} \sum_{i=1}^{d-1} i(1 - 2i - 2(i+m)_d) \gamma^2 + \cO(\gamma^3) \right)  \delta_{m,n} \label{eq:1_damp}
\end{align}
\item The single-qudit two-damping error operator $A_{2000}$ ($x=2$)
\begin{align}
    \bra{m_L}A_{2000}^{\dag}A_{20000}\ket{n_L} &= \left( \frac{1}{2d} \sum_{i=2}^{d-1}  i(i-1) \gamma^2 + \cO(\gamma^3) \right)  \delta_{m,n} \nonumber \\
    &= \left( \frac{1}{6} (d-1)(d-2) \gamma^2 + \cO(\gamma^3) \right)  \delta_{m,n} \label{eq:2_damp}
\end{align}
\item For single-qudit errors with a higher degree of damping, that is, $A_{x000}$ with $x \geq 3$, we get, 
\[\bra{m_L} A_{x000}^{\dag} A_{x0000} \ket{n_L} = \cO(\gamma^x) \delta_{m,n}.\]
\end{itemize}
 
Therefore, we see that the coefficient of $\gamma$ is independent of $m$, whereas the coefficient of $\gamma^2$ depends on the value of $m$ for $x \leq 2$. This shows that the codewords in Eq.~\eqref{eq:4qdt_code} satisfy %the unitarity (or deformability) part of 
the Knill-Laflamme conditions up to $\cO(\gamma)$ for all the single-qudit damping errors.

Finally, we consider the two-qudit errors $\{A_{1010}, A_{1001}, A_{0101}, A_{0110}\}$ in the set $\cA_{\rm corr}$. The action of these errors on the codewords is given by,
\begin{align}
    A_{1010} \ket{m_L} &= \frac{1}{\sqrt{d}} \sum_{i=1}^{d-1} \Theta[(i+m)_d-1] i (i+m)_d (1-\gamma)^{i+(i+m)_d} \left(\frac{\gamma}{1-\gamma}\right) \ket{i-1}_1 \ket{i}_2 \ket{(i+m)_d-1}_3 \ket{(i+m)_d}_4 \label{eq:3b1}\\
    A_{1001} \ket{m_L} &= \frac{1}{\sqrt{d}} \sum_{i=1}^{d-1} \Theta[(i+m)_d-1] i (i+m)_d (1-\gamma)^{i+(i+m)_d} \left(\frac{\gamma}{1-\gamma}\right) \ket{i-1}_1 \ket{i}_2 \ket{(i+m)_d}_3 \ket{(i+m)_d-1}_4 \label{eq:3b2}\\
    A_{0101} \ket{m_L} &= \frac{1}{\sqrt{d}} \sum_{i=1}^{d-1} \Theta[(i+m)_d-1] i (i+m)_d (1-\gamma)^{i+(i+m)_d} \left(\frac{\gamma}{1-\gamma}\right) \ket{i}_1 \ket{i-1}_2 \ket{(i+m)_d}_3 \ket{(i+m)_d-1}_4 \label{eq:3b3}\\
    A_{0110} \ket{m_L} &= \frac{1}{\sqrt{d}} \sum_{i=1}^{d-1} \Theta[(i+m)_d-1] i (i+m)_d (1-\gamma)^{i+(i+m)_d} \left(\frac{\gamma}{1-\gamma}\right) \ket{i}_1 \ket{i-1}_2 \ket{(i+m)_d-1}_3 \ket{(i+m)_d}_4 \label{eq:3b4}
\end{align}
Here, $\Theta[x]$ is the Heaviside step function that takes value 1 when $x \geq 0$ and 0 otherwise.   
It is evident that all the noisy states in Eqs. \eqref{eq:3b1}-~\eqref{eq:3b4} are orthogonal to each other. Furthermore, it is easy to check that the noisy states in Eqs.~\eqref{eq:3b1}-~\eqref{eq:3b4} are also orthogonal to the noisy states in Eqs.~\eqref{eq:3aa1}-~\eqref{eq:3aa4} obtained due to the single-qudit errors.

Now, we consider the following inner product, where $A_{1010}$ acts on different logical states.
\begin{align} \label{eq:3b5}
    \bra{m_L} A_{1010}^{\dag} A_{1010} \ket{n_L} = \frac{1}{d} \sum_{i=1}^{d-1} \Theta[(i+m)_d -1] i^2 \gamma^2 (i+m)_d^2 (1-\gamma)^{2i + 2(i+m)_d -2} \delta_{m,n}
\end{align}
We get identical expressions for the other two-qudit errors as well. The $\delta_{m,n}$ term in Eq.~\eqref{eq:3b5} verifies that the codewords get mapped to orthogonal states by different two-qudit error operators. The leading order term in Eq.~\eqref{eq:3b5} is $\cO(\gamma^2)$, and its coefficient does depend on the codewords.

\section {Ququad code for Amplitude-damping} \label{appendix:a1}

When the local dimension of the codewords in Eq.~\eqref{eq:4qdt_code} is four, the primary syndromes $(p_{1},p_{2})$ obtained by measuring the stablilizer generators in Eq.~\eqref{eq:stab_gen} alone cannot differentiate between the pair of errors  $(A_{2000}, A_{0200})$ and the pair of errors $(A_{0020}, A_{0002})$. However, it is possible to distinguish between these errors by utilising secondary syndrome measurements, as in the case of $d=3$.

We measure the secondary syndromes $(s_{1}, s_{2})$ for the $4$-ququad code by using the additional operators $W_4III$ and $IIW_4I$.
The matrix representation of the operator $W_4$ is given in Eq. \eqref{zhat}.
To execute these measurements, we attach an ancilla qubit and perform the following controlled operation. 

\begin{align}
     C_4=\begin{bmatrix}
    \mathbf{I}_2 & 0 & 0 & 0\\
    0 & \mathbf{I}_2 & 0 & 0\\ 
    0 &  0 & \mathbf{X}_2& 0\\
    0 &  0 & 0 & \mathbf{X}_2\end{bmatrix}\label{cx4}.
\end{align}

Here, $\mathbf{X}_2$ is the Pauli $X$ operator for qubits. The controlled operation $C_4$ flips the ancilla target qubit whenever the control state is $\ket{2}$ or $\ket{3}$.

From Table \ref{t21}, we note that the secondary syndromes $s_1, s_2$ distinguish the errors whose primary syndromes $p_1, p_2$ are identical.
\begin{table}[ht]
    \centering
   \begin{tabular}{ |p{1.75cm}||p{2.1cm}|p{2.1cm}| }
 \hline
\multicolumn{3}{|c|}{Table of syndromes for ququads }\\
\hline
 
 Errors & $\qquad p_1, p_2$ & $\qquad s_1, s_2$ \\
 \hline
 $A_{1000}$   &  3,0  &  $\times$     \\
 $A_{0100}$   &  1,0  &  $\times$     \\
 $A_{2000}$   &  2,0  &  0,0  \, {\rm or} (0,2)   \\
 $A_{0200}$   &  2,0  &  1,0  \, {\rm or} (1,2)  \\
 $A_{0010}$   &  0,3  &  $\times$      \\
 $A_{0001}$   &  0,1  &  $\times$      \\
 $A_{0020}$   &  0,2  &  0,0   \, {\rm or} (2,0)    \\
 $A_{0002}$   &  0,2  &  0,1   \, {\rm or} (2,1)    \\
 $A_{1010}$   &  3,3   & $\times$        \\
 $A_{1001}$   &  3,1   & $\times$       \\
 $A_{0101}$   &  1,1   & $\times$        \\
 $A_{0110}$   &  1,3   & $\times$        \\
 \hline
\end{tabular}
    \caption{Primary and secondary syndrome string for different errors with probability up to $\cO(\gamma^2)$ for the ququad code.}
    \label{t21}
\end{table}

\end{widetext}

\end{document}